\documentclass[%
 reprint,
 aps,
 pra,
]{revtex4-2}

\usepackage{dcolumn}
\usepackage{bm}
\usepackage{physics}
\usepackage{mathrsfs}
\usepackage{amsthm,amsmath,amssymb}
\usepackage{xspace}
\usepackage[tablename=Procedure]{caption}

\newtheorem{lemma}{Lemma}
\newtheorem{theorem}{Theorem}
\newtheorem{corollary}{Corollary}
\newcommand{\iid}{i.i.d.\@{}\xspace}
\DeclareMathOperator*{\bigcircle}{\bigcirc} 

\begin{document}

\title{One-Shot Min-Entropy Calculation Of Classical-Quantum States And Its Application To Quantum Cryptography}
\author{Rong Wang}
\email{ericrongwang19@gmail.com}
\author{H. F. Chau}
\email{hfchau@hku.hk}
\affiliation{Department of Physics, University of Hong Kong, Hong Kong SAR, China}

\date{\today}

\begin{abstract}
In quantum Shannon theory, various kinds of quantum entropies are used to characterize the capacities of noisy physical systems.  Among them, min-entropy and its smooth version attract wide interest especially in the field of quantum cryptography as they can be used to bound the information obtained by an adversary. However, calculating the exact value or non-trivial bounds of min-entropy are extremely difficult because the composite system dimension may scale exponentially with the dimension of its subsystem.  Here, we develop a one-shot lower bound calculation technique for the min-entropy of a classical-quantum state that is applicable to both finite and infinite dimensional reduced quantum states.  Moreover, we show our technique is of practical interest in at least three situations.  First, it offers an alternative tight finite-data analysis for the BB84 quantum key distribution scheme.  Second, it gives the best finite-key bound known to date for a variant of device independent quantum key distribution protocol.  Third, it provides a security proof for a novel source-independent continuous-variable quantum random number generation protocol.  These results show the effectiveness and wide applicability of our approach.
\end{abstract}

\maketitle

\onecolumngrid
\section{\label{sec:level1}Introduction}

Quantum Shannon theory \cite{wilde2013quantum} is an active subfield of quantum information processing whose aim is to quantitatively characterize the ultimate capacity of noisy physical systems.  Under the independent and identically distributed (\iid) assumption and in the asymptotic limit (that is, when there are infinitely many identical copies of the state), the relevant entropy measures are the von Neumann entropy and its variations. The situation is different in the non-asymptotic or non-\iid setting.  Here, more general entropy measures have to be used \cite{konig2009operational}.  One of them is the quantum conditional min-entropy, which we shall simply call min-entropy throughout this paper, together with its smooth version.
For example, if an adversary tries to guess a string of random variable conditioned on some accessible quantum states, then the maximum possible correctly guessing probability is the min-entropy of the random variable conditioned on the quantum states \cite{konig2009operational}.  Clearly, this adversarial setting is important as it includes important primitives such as quantum key distribution (QKD) \cite{BB84} and quantum random number generation (QRNG) \cite{herrero2017quantum}.

The min-entropy of raw data conditioned on adversary's state plays an important role in the information security analysis of quantum cryptography.  In fact, by the quantum leftover hashing lemma \cite{renner2008security,tomamichel2011leftover}, the smooth min-entropy \cite{renner2008security} determines the length of distillable key.  To calculate the lower bound of smooth min-entropy of the large composite system shared by the users and Eve, the usual way is to reduce it to \iid states and then sum up their von Neumann entropies via de Finetti representation theorem \cite{renner2007symmetry}, post-selection technique \cite{christandl2009postselection}, quantum asymptotic equipartition property \cite{tomamichel2009fully} or entropy accumulation theorem \cite{dupuis2020entropy,arnon2018practical}. However, these approaches cannot provide a tight length of final key when comparing to one-shot calculations using uncertainty relations for smooth entropies \cite{tomamichel2011uncertainty,tomamichel2012tight,tomamichel2017largely,wang2021tight}. Since these entropic uncertainty relations stem from the fact that the result of incompatible measurements are impossible to predict \cite{coles2017entropic}, their applications in quantum cryptography are limited as they require characterized measurements \cite{tomamichel2011uncertainty,tomamichel2012tight,tomamichel2017largely}. Therefore, it is instructive to find other one-shot approaches that are irrelevant to incompatible measurements.

In this work, we propose an alternative one-shot approach to min-entropy lower bound calculation, which can be extended naturally to the case of smooth min-entropy. In essence, given a classical-quantum (CQ) state whose quantum subsystem may be finite- or infinite-dimensional, we develop a technique to calculate the lower bound of min-entropy of its classical random variable conditioned on its quantum subsystem. Unlike the entropic uncertainty relation mentioned above, our approach can directly calculate the lower bound of min-entropy once the density matrix of a CQ state is given. Concretely, we can always assume that the classical variable of a CQ state is uniformly distributed in an adversarial setting.  If not, we apply Weyl operators to transform the distribution to a uniform one.  Clearly, this transformation does not decrease the adversary's information gain.  By working on the CQ state whose classical random variable is uniformly distributed, the min-entropy can be written in terms of the eigenvalues of adversary's state.  In this way, we reduce the complicated problem of min-entropy calculation to a simple problem of solving eigenvalues.

To illustrate the effectiveness of our one-shot min-entropy approach in solving a wide range of quantum cryptographic problems, we consider the following three applications.  First, we reproduce the best known provably secure key rate of the BB84 QKD protocol.  Next, we report the best provably secure key rate for a variation of Device-Independent (DI) QKD protocol under the assumption of independent measurement.  Finally, we show the security of a novel source-independent continuous-variable (SI-CV) QRNG protocol against general attacks.

\section{Definitions And Main Results}

Here we describe the task of min-entropy calculation in a quantum cryptographic setting. Alice holds a classical random variable $X$ that may take on values in the finite set $\{0, 1, \cdots, d-1\}$, while Eve holds her system $E$. Given a CQ state as  
\begin{equation}
\tau_{XE}=\sum_{x=0}^{d-1} p_x \ket{x}\bra{x} \otimes  \tau_{x}
\end{equation}
where $p_x$ is its probability distribution of $X$, $ \tau_{x}$ is Eve's state when $X$ takes the value $x$, Eve wants to maximize her chance of correctly guessing $X$ with the help of the quantum state in her system $E$.  This optimized guessing probability is given by
\begin{equation}
\label{p_guess}
p_{\text{guess}}(X|E)_{\tau_{XE}}:=\sup_{\{M_x\}} \sum_{x=0}^{d-1} p_x \Tr[M_x\tau_{x}]
\end{equation}
where the supremum is over all possible positive operator-valued measures (POVMs) $\{M_x\}$ on Eve's system.

The optimized guessing probability is related to the min-entropy of $\tau_{XE}$.
Given $\rho_{AB}$, recall that the min-entropy is defined by \cite{renner2008security,konig2009operational,tomamichel2010duality,furrer2011min,muller2013quantum}
\begin{equation}
H_{\text{min}}(A|B)_{\rho_{AB}}:=\sup_{\sigma_B} \{\lambda \in \mathbb{R}: 2^{-\lambda} \mathit{I}_A \otimes \sigma_B \ge \rho_{AB}\},
\end{equation}
where the supremum is over all normalized states $\sigma_B$ within subsystem $B$, and $\mathit{I}_A$ is the identity matrix of subsystem~A. (Throughout this paper, identity matrix of any subsystem is similarly defined.)  The max-entropy of any bipartite density matrix $\rho_{AB}$ is defined by \cite{konig2009operational,tomamichel2010duality,furrer2011min,muller2013quantum}
\begin{equation}
\label{Hmax}
H_{\text{max}}(A|B)_{\rho_{AB}}:=\sup_{\sigma_B} \,2\log F(\rho_{AB}, \mathit{I}_A \otimes \sigma_B),
\end{equation}
where $F(\cdot, \cdot)$ is the fidelity between two density matrices.  Moreover, the base of all logarithms denoted by the symbol $\log$ in this paper is $2$. The min- and max-entropies can be smoothed by introducing a family of sub-normalized state $\bar{\rho}_{AB}$ that are $\varepsilon$-close to $\rho_{AB}$ \cite{tomamichel2010duality}.  Precisely,
\begin{equation}
B_{\varepsilon}(\rho_{AB}):=\left\{\bar{\rho}_{AB}: \sqrt{ 1-F(\rho_{AB},\bar{\rho}_{AB})^2} \le \varepsilon \right\}
\end{equation}
is the set of all states that are $\varepsilon$-close to $\rho_{AB}$ and the smooth min- and max-entropies are given by
\begin{subequations}
\begin{align}
H_{\text{min}}^{\varepsilon}(A|B)_{\rho_{AB}} &:= \sup_{\bar{\rho}_{AB} \in B_{\varepsilon}(\rho_{AB})} H_{\text{min}}(A|B)_{\bar{\rho}_{AB}} , \\
H_{\text{max}}^{\varepsilon}(A|B)_{\rho_{AB}} &:= \inf_{\bar{\rho}_{AB} \in B_{\varepsilon}(\rho_{AB})} H_{\text{max}}(A|B)_{\bar{\rho}_{AB}}.
\end{align}
\end{subequations}

From the operational meaning of min-entropy \cite{konig2009operational}, the guessing probability is determined by the min-entropy of $X$ conditioned on $E$, that is 
\begin{equation}
\label{p_guess_min_entropy}
p_{\text{guess}}(X|E)_{\tau_{XE}}= 2^{-H_{\min}(X|E)_{\tau_{XE}}} .
\end{equation}
Therefore, in the cryptographic setting, the task is to bound $p_{\text{guess}}(X|E)$ or equivalently $H_{\text{min}}(X|E)_{\tau_{XE}}$.

\begin{lemma}
\label{Lem:correspondence}
For any state $\tau_{XE}$ shared between Alice and Eve, there exists a corresponding
\begin{equation}
\label{rho_XE}
\rho_{XE}=\frac{1}{d} \sum_{x=0}^{d-1}  \ket{x}\bra{x} \otimes \ket{\Psi_x}\bra{\Psi_x}
\end{equation}
with
\begin{equation}
H_{\min}(X|E)_{\rho_{XE}} \le H_{\min}(X|E)_{\tau_{XE}}.
\end{equation}
Here, $\ket{\Psi_x}$ is Eve's pure state when $X$ takes the value $x$.
Moreover, we can write $\ket{\Psi_x}$ as 
\begin{equation}
\ket{\Psi_x}= \sum_{y=0}^{d-1} \omega^{xy} \sqrt{\lambda_y} \ket{e_y},
\end{equation}
where $\omega$ is a primitive $d$-th root of unity, $\lambda_y$'s are the eigenvalues of Eve's state.  In other words,
\begin{equation}
\rho_E=\Tr_X[\rho_{XE}]=\frac{1}{d} \sum_{x=0}^{d-1} \ket{\Psi_x}\bra{\Psi_x}=\sum_{y=0}^{d-1} \lambda_y \ket{e_y}\bra{e_y} ,
\end{equation}
with $\sum_{y=0}^{d-1} \lambda_y=1$, and $\{ \ket{e_y} \}$ is an eigenbasis of subsystem $E$.
\end{lemma}

\par \medskip

\begin{proof}
By purifying each $\tau_x$ in Eve's system $E$ as $\ket{\psi_x}$ in the composed system $EF$, we obtain the state
\begin{equation}
\tau_{XEF}=\sum_{x=0}^{d-1} p_x \ket{x}\bra{x} \otimes \ket{\psi_x}\bra{\psi_x},
\end{equation}
on the enlarged system ${XEF}$.  From the data-processing inequality for min-entropy \cite{tomamichel2010duality}, we have
\begin{equation}
\label{tau_XEF}
H_{\text{min}}(X|EF)_{\tau_{XEF}} \le H_{\text{min}}(X|E)_{\tau_{XE}}.
\end{equation}
In the case of quantum cryptography, we usually assume that Eve can access to infinite computational resources. Therefore we have the liberty to assume that Eve holds the purification of $\tau_{x}$ for each $x$. Next, we introduce a symmetrical operation using the Weyl operators. Let $\{\ket{x}\}$ be a computational basis of system~A, then the Weyl operators are defined by 
\begin{equation}
\label{Uyz}
U_{yz}:=\sum_{x=0}^{d-1} \omega^{xz} \ket{x + y}\bra{x},
\end{equation}
where $y, z \in \{0,1,\cdots,d-1\}$ and summation inside the state ket is performed modulo $d$.
(From now on, all arithmetical operations inside state kets and related indices are
performed modulo $d$.)
Further define
\begin{equation}
\tau_{XEFG} := \frac{1}{d} \sum_{y=0}^{d-1}  (U_{y0} \otimes \mathit{I}_{EF}) \tau_{XEF} (U^{\dagger} _{y0} \otimes \mathit{I}_{EF})\otimes \ket{g_y} \bra{g_y},
\end{equation}
where $\{\ket{g_y}\}$ is a set of eigenbasis within Eve's subsystem $G$ whose aim is to record the information of $y$.  This symmetrical operation can be understood as follows.  Alice uniformly and randomly chooses a value $y$, applies $U_{y0}$ to her classical system.  She then publicly announces $y$ so that Eve can record this information using the register $G$.  From Appendix~A.3 in Ref. \cite{tomamichel2010duality}, this operation does not change the min-entropy. Using this trick, we obtain in a uniformly distributed random variable $X$.  This makes our later calculation easier.  The closure property of modulo $d$ allows us to write $\tau_{XEFG}$ as
\begin{equation}
\label{tau_XEFG}
\tau_{XEFG}= \frac{1}{d} \sum_{x=0}^{d-1} \ket{x}\bra{x} \otimes \left(\sum_{y=0}^{d-1} p_y \ket{\psi_y}\bra{\psi_y} \otimes \ket{g_{x - y}} \bra{g_{x - y}} \right) .
\end{equation}
Given the state
\begin{equation}
\rho_{XEFG}= \frac{1}{d} \sum_{x=0}^{d-1} \ket{x}\bra{x} \otimes P \left\{ \sum_{y=0}^{d-1} \sqrt{p_y} \ket{\psi_y} \otimes \ket{g_{x - y}} \right\},
\end{equation}
where $P\{\ket{x}\}:=\ket{x}\bra{x}$, it is straightforward to see that 
\begin{equation}
\tau_{XEFG}= \sum_{y=0}^{d-1} \ket{g_y}\bra{g_y} \rho_{XEFG}  \ket{g_y}\bra{g_y},
\end{equation}
where the set of projectors $\{\ket{g_y}\bra{g_y}\}$ determines a projective measurement. Using the data-processing inequality for min-entropy in Sec. V of Ref. \cite{tomamichel2010duality}, we have 
\begin{equation}
H_{\text{min}}(X|EFG)_{\rho_{XEFG}} \le H_{\text{min}}(X|EFG)_{\tau_{XEFG}}.
\end{equation}
 We now pick a mutually unbiased basis $\{\ket{h_y}\}_{y=0}^{d-1}$ where
\begin{equation}
\ket{h_y}= \frac{1}{\sqrt{d}} \sum_{z=0}^{d-1} \omega^{-yz} \ket{g_z} .
\end{equation}
Then $\rho_{XEFG}$ can be rewritten as
\begin{equation}
\rho_{XEFG} = \frac{1}{d} \sum_{x=0}^{d-1} \ket{x}\bra{x} \otimes P \left\{ \sum_{y=0}^{d-1} \omega^{xy} \left( \frac{1}{\sqrt{d}} \sum_{z=0}^{d-1} \omega^{-yz} \sqrt{p_z} \ket{\psi_z} \right) \otimes \ket{h_y} \right\},
\end{equation}
where we also make use of the closure property of modulo $d$. Note that 
\begin{equation}
0 \le \left| \frac{1}{\sqrt{d}} \sum_{z=0}^{d-1} \omega^{-yz} \sqrt{p_z} \ket{\psi_z} \right|^2 \le \frac{1}{d} \left(\sum_{z=0}^{d-1} \sqrt{p_z} \right)^2 \le 1,
\end{equation}
where we have used the Cauchy–Schwarz inequality to arrive at the last inequality. Finally, we set 
\begin{subequations}
\begin{align}
\lambda_y &=\left| \frac{1}{\sqrt{d}} \sum_{z=0}^{d-1} \omega^{-yz} \sqrt{p_z} \ket{\psi_z} \right|^2 , \\
\sqrt{\lambda_y} \ket{e_y} &= \left( \frac{1}{\sqrt{d}} \sum_{z=0}^{d-1} \omega^{-yz} \sqrt{p_z} \ket{\psi_z} \right) \otimes \ket{h_y} , \\
\ket{\Psi_x} &= \sum_{y=0}^{d-1} \omega^{xy} \left( \frac{1}{\sqrt{d}} \sum_{z=0}^{d-1} \omega^{-yz} \sqrt{p_z} \ket{\psi_z} \right) \otimes \ket{h_y} ,
\end{align}
\end{subequations}
and rephrase Eve's systems $EFG$ using one system $E$. By combining Eqs. \eqref{tau_XEF}, \eqref{tau_XEFG} and the setting above, we obtain $H_{\text{min}}(X|E)_{\rho_{XE}} \le H_{\text{min}}(X|E)_{\tau_{XE}}$.  This completes our proof.
\end{proof}

\begin{theorem}
\label{Thrm:H_min}
The min-entropy of the state $\rho_{XE}$ given by Eq. \eqref{rho_XE} in Lemma~\ref{Lem:correspondence} equals
\begin{equation}
\label{min-entropy}
H_{\min}(X|E)_{\rho_{XE}}=\log d-\log (\sum_{y=0}^{d-1} \sqrt{\lambda_y})^2.
\end{equation}
 Hence,
\begin{equation}
 H_{\min}(X|E)_{\tau_{XE}} \ge \log d-\log (\sum_{y=0}^{d-1} \sqrt{\lambda_y})^2.
\end{equation}
\end{theorem}

\begin{proof}
 From Lemma~\ref{Lem:correspondence}, it suffices to prove the validity of Eq.~\eqref{min-entropy}.
By purifying $\rho_{XE}$ with another system $X'$, an alternative purification is given by 
\begin{equation}
\ket{\Phi}_{XX'E}=\sum_{y=0}^{d-1} \sqrt{\lambda_y} \left(\frac{1}{\sqrt{d}} \sum_{x=0}^{d-1} \omega^{xy} \ket{xx} \right) \ket{e_y},
\end{equation}
where we extend the definition of $X$ from a classical random variable to a quantum system and define the state $\ket{xx}$ within composite system $XX'$. Owing to the duality between min-entropy and max-entropy \cite{tomamichel2010duality}, we have 
\begin{equation}
H_{\text{min}}(X|E)_{\rho_{XE}}+H_{\text{max}}(X|X')_{\rho_{XX'}}=0,
\end{equation}
where 
\begin{equation}
\rho_{XX'}=\Tr_E\Big[\ket{\Phi}_{XX'E}\bra{\Phi}\Big]=\sum_{y=0}^{d-1} \lambda_y P \left\{\frac{1}{\sqrt{d}} \sum_{x=0}^{d-1} \omega^{xy} \ket{xx} \right\}
\end{equation}
is the state after partially tracing system $E$ out. Note that $\sum_{x=0}^{d-1} \omega^{xy} \ket{xx} / \sqrt{d}$ is the generalized Bell state for each $y \in \{0, 1, \cdots, d-1\}$, and therefore $\rho_{XX'}$ is Bell-diagonal. In what follows, we are going to prove that 
\begin{equation}
H_{\text{max}}(X|X')_{\rho_{XX'}}=\sup_{\sigma_{X'}} \,2\log F(\rho_{XX'}, \mathit{I}_X \otimes \sigma_{X'})= \log \left[ \frac{1}{d} \left(\sum_{y=0}^{d-1} \sqrt{\lambda_y} \right)^2 \right],
\end{equation}
where $\sigma_{X'}$ would be optimized over all states within the system $X'$. Let $\sigma_{X'}=\mathit{I}_{X'}/d$, then 
\begin{equation}
H_{\text{max}}(X|X')_{\rho_{XX'}} \ge 2\log F(\rho_{XX'}, \mathit{I}_X \otimes \frac{\mathit{I}_{X'}}{d}) = \log \left[ \frac{1}{d} \left(\sum_{y=0}^{d-1} \sqrt{\lambda_y} \right)^2 \right].
\end{equation}
Let $\sigma$ be the state that 
\begin{equation}
H_{\text{max}}(X|X')_{\rho_{XX'}} = 2\log F(\rho_{XX'}, \mathit{I}_X \otimes \sigma), 
\end{equation}
we can complete the proof by showing
\begin{equation}
F(\rho_{XX'}, \mathit{I}_X \otimes \sigma) \le F(\rho_{XX'}, \mathit{I}_X \otimes \frac{\mathit{I}_{X'}}{d}) .
\end{equation}
With the Weyl operators defined in Eq. \eqref{Uyz}, we apply $U^*_{yz} \otimes U_{yz}$ to the states in system $XX'$ to get
\begin{equation}
F(\rho_{XX'}, \mathit{I}_X \otimes \sigma)=F\Big((U^*_{yz} \otimes U_{yz})\rho_{XX'}(U^*_{yz} \otimes U_{yz})^{\dagger}, (U^*_{yz} \otimes U_{yz})\mathit{I}_X \otimes \sigma(U^*_{yz} \otimes U_{yz})^{\dagger}\Big) = F(\rho_{XX'}, \mathit{I}_X \otimes U_{yz} \sigma U^{\dagger}_{yz}).
\end{equation}
Here, we make use of the properties that fidelity is invariant under unitary transformations and that the generalized Bell states are the eigenstates of all the $U^*_{yz} \otimes U_{yz}$. As fidelity is concave, we have 
\begin{equation}
F(\rho_{XX'}, \mathit{I}_X \otimes \sigma) = \frac{1}{d^2} \sum_{y,z=0}^{d-1} F(\rho_{XX'}, \mathit{I}_X \otimes U_{yz} \sigma U^{\dagger}_{yz}) \le F(\rho_{XX'}, \mathit{I}_X \otimes  \frac{1}{d^2} \sum_{y,z=0}^{d-1} U_{yz} \sigma U^{\dagger}_{yz}) =  F(\rho_{XX'}, \mathit{I}_X \otimes \frac{\mathit{I}_{X'}}{d}).
\end{equation}
Here we make use of the fact that applying $\{U_{yz}\}$ uniformly on any state will result in 
the unique maximally mixed state.  This proves our theorem.
\end{proof}

\begin{corollary} 
\label{Cor:min_ent_for_rho_XE}
 For the state $\rho_{XE}$ given by Eq. \eqref{rho_XE} in Lemma~\ref{Lem:correspondence} and its min-entropy expressed in Eq. \eqref{min-entropy} in Theorem~\ref{Thrm:H_min}, the guessing probability $p_{\text{guess}}(X|E)_{\rho_{XE}}$ can be attained by the POVM $\{ M_x \}_{x=0}^{d-1}$ where
\begin{equation}
 \label{POVM_M_x}
M_x=P \left \{ \frac{1}{\sqrt{d}} \sum_{y=0}^{d-1} \omega^{xy} \ket{e_y} \right\}.
\end{equation}
\end{corollary}

\begin{proof}
 According to the definition in Eq. \eqref{p_guess_min_entropy}, we have
\begin{equation}
p_{\text{guess}}(X|E)_{\rho_{XE}}= \frac{1}{d}\left(\sum_{y=0}^{d-1} \sqrt{\lambda_y} \right)^2 =
\frac{1}{d} \sum_{x=0}^{d-1} \left|\sum_{y=0}^{d-1} \omega^{-xy} \bra{e_y}  \ket{\Psi_x} \right|^2.
\end{equation}
This proves the corollary.
\end{proof}

\section{Application To Quantum Cryptography}

\subsection{Application To Discrete Variable Quantum Key Distribution}

We now demonstrate how to apply our approach to the finite-data security analysis in DV QKD protocol, using BB84 protocol as an example.  For completeness, we briefly describe the EB version of BB84 protocol in Procedure~\ref{procedure:QKD}.

\begin{table}
 \caption{\label{procedure:QKD}The Entanglement-Based BB84 Protocol}
 \begin{enumerate}
  \item \emph{State preparation:} Alice prepares a pair of maximally entangled qubit $(\ket{00}+\ket{11})/\sqrt{2}$.
   \label{procedure:QKD:preparation}
  \item \emph{State distribution:} Alice sends the second qubit of state $(\ket{00}+\ket{11})/\sqrt{2}$ to Bob through an insecure channel.
  \item \emph{Detection:} Bob publicly announces whether he detects the qubit or not.  Alice keeps her qubit only if Bob successfully detects the qubit.
   \label{procedure:QKD:detection}
  \item \emph{Measurement:} After repeating steps~\ref{procedure:QKD:preparation}--\ref{procedure:QKD:detection} many times, Alice (Bob) randomly and independently chooses a basis $\mathbb{Z}$ or $\mathbb{X}$ to measure her (his) each of their qubit in hand.
  \item \emph{Sifting:} Alice (Bob) publicly announces her (his) basis information.  They keep only the basis-matched data.  Denote the number of qubits that they both measured in $\mathbb{X}$ and $\mathbb{Z}$ bases by $(n+k)$ and $k$, respectively.
  \item \emph{Parameter estimation:} Alice and Bob randomly disclose $k$ of the $(n+k)$ bits of data from their $\mathbb{X}$ measurement results to compute the bit error frequency $e_x$.  They also announce all the $k$ bits of data from their $\mathbb{Z}$ measurement to compute the bit error frequency $e_z$. They continue only if both of $e_x$ and $e_z$ do not exceed a predefined threshold.
  \item \emph{Error correction:} For the remaining $n$ bits of data from their $\mathbb{X}$ basis measurement, Alice and Bob execute an information reconciliation scheme that leaks at most $\text{leak}_{\text{EC}}+\left \lceil \log_{2}\frac{1}{\varepsilon_\text{cor}} \right \rceil$ bits if the protocol is $\varepsilon_\text{cor}$-correct.
   \label{procedure:QKD:EC}
  \item \emph{Privacy amplification:} They apply a random two-universal hash function to the error-corrected bits in step~\ref{procedure:QKD:EC} to extract $\ell$ bits of secret key.
 \end{enumerate}
\end{table}

According to the quantum leftover hashing lemma \cite{tomamichel2011leftover}, the extracted secure key length is determined by the smooth min-entropy of raw key conditioned on Eve's quantum side information. Therefore, the core issue is to express this smooth min-entropy in terms of the observable statistics.

Suppose we were to measure the $n$-round’s key generation data resulted from performing $\mathbb{X}$ measurement on $n$-pair qubit in the $\mathbb{Z}$ basis.  We may estimate the upper bound of frequency of bit error denoted as $\hat{e}_z$ from $e_z$ as follows. Picking any concentration inequality without replacement, such as the Serfling inequality \cite{tomamichel2012tight}, the upper bound $\hat{e}_z$ can be written as a function of parameters $e_z$, $m$, $k$ and the small failure probability $\varepsilon^2$, namely,
\begin{equation}
\hat{e}_z=e_z+\sqrt{\frac{(n+k)(k+1)}{nk^2}\ln \frac{1}{\varepsilon}}.
 \label{E:hat_e}
\end{equation}
(Note that Eq.~\eqref{E:hat_e} depends on the concentration inequality used.  However, the detailed form of Eq.~\eqref{E:hat_e} is not the main focus of this paper.)
In what follows, we express smooth min-entropy in terms of $\hat{e}_z$, with $\varepsilon$ taken as the smooth parameter.

Following the reduction idea in Lo-Chau proof  \cite{lo1999unconditional} and similar techniques \cite{kraus2005lower,renner2005information,renner2008security}, we reduce this $n$-pair of qubit to the generalized Bell-diagonal one. To express this, we define two $n$-bit strings $\bm{i}, \bm{j} \in \{0, 1\}^n$, and denote the $k$-th bits of the bit strings of $\bm{i}$ and $\bm{j}$ by $\bm{i}_k, \bm{j}_k \in \{0, 1\}$, respectively.  Let $\sigma_{0, 0}:=\mathit{I}$ be the two-dimensional identity matrix, and $\sigma_{1, 0}:=\sigma_z$, $\sigma_{0, 1}:=\sigma_x$, $\sigma_{1, 1}:=\sigma_y$ to be the three Pauli matrices. Further define
\begin{equation}
\label{U_ij}
U_{\bm{i}, \bm{j}}:=\bigotimes_{k=1}^{n} \sigma_{\bm{i}_k, \bm{j}_k} .
\end{equation}
Then, the $n$-pair Bell state is thus given by
\begin{equation}
\ket{\Phi_{\bm{i}, \bm{j}} }:= \left(\mathit{I}^{\otimes n}_A \otimes U_{\bm{i}, \bm{j}} \right)\left( \frac{\ket{00}+\ket{11}}{\sqrt{2}} \right)^{\otimes n}.
\end{equation}
As a consequence, this $n$-pair Bell-diagonal state shared by Alice and Bob is given by
\begin{equation}
\label{rhoabn}
\begin{aligned}
\rho^n_{AB}:=\sum_{\bm{i}, \bm{j}} \lambda_{\bm{i}, \bm{j}} \ket{\Phi_{\bm{i}, \bm{j}} } \bra{\Phi_{\bm{i}, \bm{j}} },
\end{aligned}
\end{equation}
where $\{\lambda_{\bm{i}, \bm{j}}\}$ are some real-valued and non-negative coefficients satisfying $\sum_{\bm{i}, \bm{j}} \lambda_{\bm{i}, \bm{j}}=1$.

\begin{theorem} \label{Thrm:Psi_ABE}
 Let Eve holds the purification of $\rho^n_{AB}$ defined in Eq. \eqref{rhoabn}, then the states shared by Alice, Bob and Eve is given by
\begin{equation}
\ket{\Psi}^n_{ABE}:=\sum_{\bm{i}, \bm{j}} \sqrt{\lambda_{\bm{i}, \bm{j}}} \ket{\Phi_{\bm{i}, \bm{j}}} \otimes \ket{e_{\bm{i}, \bm{j}}}.
\end{equation}
Here, $\ket{e_{\bm{i}, \bm{j}}}$'s are mutually orthogonal within Eve's system $E$. Let $\rho^n_{XE}$ be the state after Alice measures all her qubits in the $\mathbb{X}$ basis and then traces out Bob's system.  Then, the min-entropy of $\rho^n_{\bm{X}E}$ equals
\begin{equation}
\label{hminrho_xe}
H_{\min}(\bm{X}|E)_{\rho^n_{\bm{X}E}}=n-\log \left[\sum_{\bm{j}} \left( \sum_{\bm{i}} \sqrt{\lambda_{\bm{i}, \bm{j}}} \right)^2 \right].
\end{equation}
\end{theorem}

\begin{proof}
Clearly, the CQ state $\rho^n_{\bm{X}E}$ is in the form
\begin{equation}
\begin{aligned}
\rho^n_{\bm{X}E}=\frac{1}{2^n} \sum_{\bm{x}} \ket{\bm{x}}\bra{\bm{x}} \otimes 
\sum_{\bm{j}}  P \left\{ \sum_{\bm{i}} (-1)^{\bm{i} \cdot \bm{x}} \sqrt{\lambda_{\bm{i}, \bm{j}}}  \ket{e_{\bm{i}, \bm{j}}}  \right\},
\end{aligned}
\end{equation}
where $\bm{x} \in \{0, 1\}^n$ is Alice's classical bit string. The sub-normalized state conditioned on $\bm{j}$ is
\begin{equation}
\rho^n_{\bm{X}E|\bm{j}}=\frac{1}{2^n} \sum_{\bm{x}} \ket{\bm{x}}\bra{\bm{x}} \otimes  P \left\{ \sum_{\bm{i}} (-1)^{\bm{i} \cdot  \bm{x}} \sqrt{\lambda_{\bm{i}, \bm{j}}}  \ket{e_{\bm{i}, \bm{j}}}  \right\},
\end{equation}
where $\rho^n_{\bm{X}E}=\sum_{\bm{j}}  \rho^n_{\bm{X}E|\bm{j}}$. Because of the mutually orthogonality of $\rho^n_{\bm{X}E|\bm{j}}$, Appendix A.3 in Ref. \cite{tomamichel2010duality} implies that
\begin{equation}
\label{hminrelation}
2^{-H_{\min}(\bm{X}|E)_{\rho^n_{\bm{X}E}}}=\sum_{\bm{j}} 2^{-H_{\min}(X|E)_{\rho^n_{\bm{X}E|\bm{j}}}}.
\end{equation}
 By applying Theorem \ref{Thrm:H_min} to $\rho^n_{\bm{X}E|\bm{j}}$ and setting $d=2^n$, we get
\begin{equation}
H_{\min}(\bm{X}|E)_{\rho^n_{\bm{X}E|\bm{j}}}=n-\log( \sum_{\bm{i}} \sqrt{\lambda_{\bm{i}, \bm{j}}})^2.
\end{equation}
Combining with Eq. \eqref{hminrelation}, we obtain Eq. \eqref{hminrho_xe}.
\end{proof}

We further simplify the form of $H_{\text{min}}(\bm{X}|E)_{\rho^n_{\bm{X}E}}$ using the following corollary.

\begin{corollary} 
\label{Cor:H_min_XE}
Let $\lambda_{\bm{i}}:=\sum_{\bm{j}} \lambda_{\bm{i}, \bm{j}}$, then
\begin{equation}
H_{\min}(\bm{X}|E)_{\rho^n_{\bm{X}E}} \ge n-\log (\sum_{\bm{i}} \sqrt{\lambda_{\bm{i}}})^2.
\end{equation}
\end{corollary}

\begin{proof}
From Eq. \eqref{hminrelation}, it suffices to prove 
\begin{equation}
\sum_{\bm{j}} \left( \sum_{\bm{i}} \sqrt{\lambda_{\bm{i}, \bm{j}}} \right)^2 \le \left(\sum_{\bm{i}} \sqrt{\lambda_{\bm{i}}} \right)^2,
\end{equation}
or equivalently 
\begin{equation}
\sum_{\bm{j}} \left( \sum_{\bm{i}}\lambda_{\bm{i}, \bm{j}}  + 2 \sum_{\bm{i} \ne \bm{i}^{'}} \sqrt{\lambda_{\bm{i}, \bm{j}}  \lambda_{\bm{i}^{'}, \bm{j}}  } \right) \le  \sum_{\bm{i}}\lambda_{\bm{i}}  + 2 \sum_{\bm{i} \ne \bm{i}^{'}} \sqrt{\lambda_{\bm{i}}  \lambda_{\bm{i}^{'}}  }.
\end{equation}
Therefore, it is sufficient to show that
\begin{equation}
 \sum_{\bm{j}} \sqrt{\lambda_{\bm{i}, \bm{j}}  \lambda_{\bm{i}^{'}\!, \bm{j}}  } \le \sqrt{\lambda_{\bm{i}}  \lambda_{\bm{i}^{'}}},
\end{equation}
for any pair of $\bm{i} \ne \bm{i}^{'}$. And this inequality follows directly from the Cauchy-Schwarz inequality.
\end{proof}

In the following lemma, we connect $\{\lambda_{\bm{i}}\}$ with $\hat{e}_z$, and thus express the smoothed version of $H_{\text{min}}(\bm{X}|E)_{\rho^n_{XE}}$ in terms of $\hat{e}_z$.

\begin{lemma} 
\label{Lem:smooth_H_min}
Let $\varepsilon \ge 0$. Define the set
\begin{equation}
\mathcal{S}_{\hat{e}_z}:=\left\{\bm{i}: \frac{\sum_k \bm{i}_k }{n} \le \hat{e}_z \right\} .
\end{equation}
Suppose ${\sum_{\bm{i} \in \mathcal{S}_{\hat{e}_z} } \lambda_{\bm{i}}}=1-\varepsilon$. Then
\begin{equation}
\label{H_min^varepsilon}
H_{\min}^{\varepsilon}(\bm{X}|E)_{\rho^n_{\bm{X}E}} \ge n[1-h(\hat{e}_z)],
\end{equation}
where $h(x):=-x\log x -(1-x)\log(1-x)$ is the binary entropy function.
\end{lemma}

\begin{proof}
We claim that the state
\begin{equation}
\ket{\Psi_{\varepsilon}}^n_{ABE}=\frac{1} {{\sum_{\bm{i} \in \mathcal{S}_{\hat{e}_z} } \lambda_{\bm{i}}}}  \sum_{\bm{i} \in \mathcal{S}_{\hat{e}_z}}  \sum_{\bm{j}} \sqrt{\lambda_{\bm{i}, \bm{j}}} \ket{\Phi_{\bm{i}, \bm{j}}} \otimes \ket{e_{\bm{i}, \bm{j}}} .
\end{equation}
is $\varepsilon$-close to $\ket{\Psi}^n_{ABE}$ in the terms of the purified distance \cite{tomamichel2010duality}. To see this, we calculate the fidelity between the two states, namely,
\begin{equation}
F(\ket{\Psi_{\varepsilon}}^n_{ABE}, \ket{\Psi}^n_{ABE})=\left|\bra{\Psi_{\varepsilon}}^n_{ABE} \ket{\Psi}^n_{ABE}\right|=\sqrt{{\sum_{\bm{i} \in \mathcal{S}_{\hat{e}_z} } \lambda_{\bm{i}}}}=\sqrt{1-\varepsilon},
\end{equation}
so that the purified distance between the two states is then given by
\begin{equation}
P(\ket{\Psi_{\varepsilon}}^n_{ABE}, \ket{\Psi}^n_{ABE})=\sqrt{1-F^2(\ket{\Psi_{\varepsilon}}^n_{ABE}, \ket{\Psi}^n_{ABE})}=\varepsilon.
\end{equation}
Assuming that $\sigma^n_{\bm{X}E}$ is the state resulted from $\ket{\Psi_{\varepsilon}}^n_{ABE}$ after Alice measures all her qubits in $\mathbb{X}$ basis and then traces out Bob's system.  From the monotonicity of the fidelity, we have 
\begin{equation}
P(\rho^n_{\bm{X}E}, \sigma^n_{\bm{X}E}) \le \varepsilon.
\end{equation}
Therefore, using Corollary \ref{Cor:H_min_XE} and the definition of smooth min-entropy \cite{tomamichel2010duality}, we obtain 
\begin{equation}
H^{\varepsilon}_{\min}(\bm{X}|E)_{\rho^n_{\bm{X}E}} \ge H_{\min}(\bm{X}|E)_{\sigma^n_{\bm{X}E}} \ge
n-\log  \frac{(\sum_{\bm{i} \in \mathcal{S}_{\hat{e}_z} } \sqrt{\lambda_{\bm{i}}})^2 } {\sum_{\bm{i} \in \mathcal{S}_{\hat{e}_z} } \lambda_{\bm{i}}} .
\end{equation}
Since the cardinality of $\mathcal{S}_{\hat{e}_z}$ is at most 
$\displaystyle \sum_{\omega=0}^{\left \lfloor  n\hat{e}_z \right \rfloor} \binom{n}{\omega}$, using the technique of Lemma 3 in the Supplementary Information of Ref. \cite{tomamichel2012tight}, we arrive at
\begin{equation}
\log  \frac{(\sum_{\bm{i} \in \mathcal{S}_{\hat{e}_z} } \sqrt{\lambda_{\bm{i}}})^2 } {\sum_{\bm{i} \in \mathcal{S}_{\hat{e}_z} } \lambda_{\bm{i}}} \le \log \sum_{\omega=0}^{\left \lfloor  n\hat{e}_z \right \rfloor} \binom{n}{\omega} \le nh(\hat{e}_z).
\end{equation}
This completes the proof.
\end{proof}

Finally, by applying the quantum leftover hashing lemma \cite{renner2008security,tomamichel2011leftover}, we obtain a lower bound for the secret key length $\ell$ and prove that the protocol is $\varepsilon_\text{sec}$-secret with $\varepsilon_\text{sec}=4\varepsilon$.  Remarkably, this bound is the same as those reported in Ref. \cite{tomamichel2012tight}.

\begin{theorem} 
\label{Thrm:keylength}
If the final key length $\ell$ obeys
\begin{equation}
\label{key_length}
 \ell \le n[1-h(\hat{e}_z)]-{\normalfont\text{leak}_{\text{EC}}}- \log\frac{2}{\varepsilon_{\normalfont\text{sec}}^2\varepsilon_{\normalfont\text{cor}}}, 
\end{equation}
 where the classical information of the error correction leaked to Eve is at most ${\normalfont\text{leak}_{\text{EC}}}+\log_{2}(1/\varepsilon_{\normalfont\text{cor}})$, then this protocol is $\varepsilon_{\normalfont\text{sec}}$-secret. 
\end{theorem}

\begin{proof}
According to quantum leftover hashing lemma \cite{renner2008security,tomamichel2011leftover}, users can extract a $\Delta$-secret key of length $\ell$ from string $\bm{X}$ where
\begin{equation}
\Delta=2\varepsilon+\frac{1}{2}\sqrt{2^{\ell-H_{\text{min}}^{\varepsilon}(\bm{X}|E')}}.
\end{equation}
Here the term $E'$ represents all information Eve obtained, including the classical information of the error correction step and Eve's quantum side information. Since at most $\text{leak}_{\text{EC}}+\log_{2}(1/\varepsilon_\text{cor})$ bits leaked to Eve in the error correction, by the chain rule for smooth min-entropy \cite{vitanov2013chain} plus Eq. \eqref{H_min^varepsilon} in Lemma \ref{Lem:smooth_H_min}, we get
\begin{equation}
H_{\text{min}}^{\varepsilon}(\bm{X}|E^{'}) \ge H_{\text{min}}^{\varepsilon}(\bm{X}|E)_{\rho^n_{\bm{X}E}}-\text{leak}_{\text{EC}}- \log_{2}\frac{2}{\varepsilon_\text{cor}} \ge n[1-h(\hat{e}_z)]-\text{leak}_{\text{EC}}- \log_{2}\frac{2}{\varepsilon_\text{cor}}.
\end{equation}
By putting $\varepsilon_\text{sec}=4\varepsilon$, we obtain
\begin{equation}
\Delta \le 2\varepsilon+\frac{1}{2}\sqrt{2^{l-H_{\text{min}}^{\varepsilon}(\bm{X}|E')}} \le \frac{\varepsilon_\text{sec}}{2}+\frac{\varepsilon_\text{sec}}{2}=\varepsilon_\text{sec}.
\end{equation}
Thus, this protocol is $\varepsilon_\text{sec}$-secret.
\end{proof}

We remark that by universal composable security \cite{renner2008security,muller2009composability,portmann2022security}, this protocol is $\varepsilon_\text{tot}=(\varepsilon_\text{sec}+\varepsilon_\text{cor})$-secure. Comparing our key length in Eq. \eqref{key_length} with that of Ref. \cite{tomamichel2012tight}, the only difference is the estimation of $\hat{e}_z$, which depends on the actual concentration inequality used. If we use the same concentration inequality, our key length expression would be the same as theirs.  In conclusion, our one-shot smooth min-entropy bound calculation is powerful enough to reproduce the best provably secure key rate of the standard BB84 scheme in the finite-data setting.

\subsection{Application To Device-Independent Quantum Protocol With Uncharacterized Measurements}

Entropic uncertainty relations may be ineffective to show the security involving uncharacterized or imperfect measurements.  To further show the significance of our approach on handling this problem, we consider the security analysis of a QKD protocol regarding to uncharacterized measurements. We pick the protocol introduced in Ref. \cite{masanes2011secure} as the example. Concretely, it is a variant of Device-Independent (DI) QKD \cite{ekert1991quantum,acin2007device,pironio2009device} with assumption of independent measurements. This independence condition may be justifiable in several implementations and is necessarily satisfied when the raw key is generated by $N$ separate pairs of devices \cite{masanes2011secure}. In a more practical implementation, in which the raw key is generated by repeatedly performing measurements in sequence on a single pair of devices, this assumption means that the functioning of the devices do not depend on any internal memory storing the quantum states and measurement results obtained in previous rounds \cite{masanes2011secure}. In other words, it is not secure against memory attack \cite{barrett2013memory}. We write down the procedure of this variant DI QKD protocol in Procedure~\ref{procedure_DIQKD}.

\begin{table}
 \caption{\label{procedure_DIQKD}Device-Independent Quantum Key Distribution Protocol With Causally Independent Measurements}
 \begin{enumerate}
  \item \emph{State preparation and distribution:} Eve prepares an $N$-pair qubit state.  She sends half of each pair to Alice and the other half to Bob.
  \item \emph{Measurement:} For each qubit pair sent by Eve, both Alice and Bob randomly and separately choose either the key generation mode or the testing mode.  If the testing mode is selected, Alice (Bob) uniformly at random chooses $\kappa_a=0,1$ ($\kappa_b=0,1$) where $\kappa_a \in \{0, 1\}$ and $\kappa_b \in \{0, 1, 2\}$.  Whereas if the key generation mode is picked, Alice (Bob) sets $\kappa_a=0$ ($\kappa_b=2$). Alice (Bob) separately performs measurements on his (her) share of the qubit pair using operator $A_{\kappa_a}$ ($B_{\kappa_b}$) defined in and near Eq.~\eqref{E:meas_opers_def}.  Here $A_i$'s and $B_i$'s are complete projective measurements for $i = 0, 1$.  Moreover, $B_2$ can be chosen to be either $A_0$ or a general POVM.  (See the discussions near Eq.~\eqref{E:meas_opers_def} for detail.)  They jot down their measurement result as the bits $x$ and $y$, respectively.
  \item \emph{Sifting:} Alice (Bob) publicly announces her (his) mode. And they keep the mode-matched data. Denote the number of testing rounds by $4k$ with each combination $(\kappa_a, \kappa_b) \in \{00, 01, 10, 11\}$ by $k$. Further denote the number of key generation mode by $4n$. Clearly, $4n+4k=N$.  (By repeating the procedure sufficiently many times and by dropping at most 6 ``rounding events'', we may assume for simplicity that $n$ and $k$ are large positive integers.)
  \item \emph{Parameter estimation:} Alice and Bob calculate the winning probability of the Clauser–Horne–Shimony–Holt (CHSH) game using their measurement results from the testing mode.  That is to say, they compute the chance that $x \oplus y=\kappa_a \cdot \kappa_b$~\cite{clauser1969proposed}.  They abort the protocol if the winning frequency does not exceed a predefined threshold.
  \item \emph{Classical Post-processing:} For the remaining $4n$ bits of data from key generation mode, Alice and Bob execute an information reconciliation scheme that leaks at most $\text{leak}_{\text{EC}}+\left \lceil \log_{2}\frac{1}{\varepsilon_\text{cor}} \right \rceil$ bits if the protocol is $\varepsilon_\text{cor}$-correct.  Then, they apply a random two-universal hash function to the resultant error-corrected bits to extract $\ell$ bits of secret key. 
 \end{enumerate}
\end{table}

We begin our security analysis by applying our one-shot approach to compute the min-entropy of a single round of key generation mode. Using the reduction techniques in Lemmas~1 and~2 of Ref.~\cite{pironio2009device}, we could reduce any measurements to (complete) PMs in this protocol.  We denote the measurement operation corresponding to Alice (Bob) choosing $\kappa_a \in \{ 0,1\}$ ($\kappa_b \in \{ 0,1,2\}$) by $A_{\kappa_a}$ ($B_{\kappa_b}$).  Without lost of generality, we write
\begin{subequations}
 \label{E:meas_opers_def}
\begin{align}
A_0 &=\cos\alpha \sigma_z + \sin\alpha \sigma_x ,\\
A_1 &=\cos\alpha \sigma_x + \sin\alpha \sigma_z ,\\
B_0 &=\cos\beta \sigma_z + \sin\beta \sigma_x ,\\
B_1 &=\cos\beta \sigma_z - \sin\beta \sigma_x , 
\end{align}
\end{subequations}
for some fixed $\alpha, \beta \in \mathbb{R}$.  In these operators, we use the convention that their eigenvectors correspond to the rank one projects used in the complete projective measurement.  Furthermore, those measurement outcomes corresponding to eigenvalue $1$ ($-1$) are assigned the bit value $1$ ($0$).  The choice of $B_2$ requires an explanation.  From the reduction techniques in Lemmas~1 and~2 of Ref.~\cite{pironio2009device}, for unbounded and countable systems, it is possible to reduce two POVM measurements to two PMs.  However, it is not clear if it can be reduced to three PMs.  Hence, in general we may just leave it as a general POVM.  Nonetheless, we stress that the form of $B_2$ does not affect our subsequent security analysis.

Using the measurement operations defined in Eq.~\eqref{E:meas_opers_def} above, in the testing mode, the CHSH operator is given by 
\begin{equation}
\begin{aligned}
S:=&\frac{A_0 \otimes B_0 + A_0 \otimes B_1 + A_1 \otimes B_0 - A_1 \otimes B_1}{4} \\  
    =&  \frac{1}{2} \left(\cos\alpha \cos\beta \sigma_z \otimes \sigma_z + \sin\alpha\cos\beta \sigma_x \otimes \sigma_z + \cos\alpha\sin\beta \sigma_x \otimes \sigma_x + \sin\alpha\sin\beta \sigma_z \otimes \sigma_x \right).	    
\end{aligned}
\end{equation}
By standard procedure of single value decomposition, the formula of CHSH operator could be rewritten as 
\begin{align}
\label{CHSH}
 S=\frac{1}{2} \left(\sqrt{\Lambda_+} \tilde{\sigma}_z \otimes \tilde{\sigma}_z + \sqrt{\Lambda_-} \tilde{\sigma}_x \otimes \tilde{\sigma}_x \right),     
\end{align}
where $\Lambda_{\pm}=\frac{1}{2}(1 \pm \sqrt{1-\cos^2 2\alpha \sin^2 2\beta} ) \ge 0$ and $\tilde{\sigma}_i$'s are the Pauli operators in a suitable locally unitarily rotated reference frame.  Clearly, for any pair of qubit $\rho_{AB}$, its CHSH value is determined by $\left \langle \tilde{\sigma}_z \otimes \tilde{\sigma}_z \right \rangle$ and $\left \langle \tilde{\sigma}_x \otimes \tilde{\sigma}_x \right \rangle$. Observe that Eve can perform Bell-state measurement before distribution, which would neither change $\left \langle \tilde{\sigma}_z \otimes \tilde{\sigma}_z \right \rangle$ nor $\left \langle \tilde{\sigma}_x \otimes \tilde{\sigma}_x \right \rangle$, nor decrease her information knowledge. Thus, we write the distributed $\rho_{AB}$ as Bell-diagonal one, given by
\begin{align}
\label{rho_AB}
\rho_{AB}= \sum_{i,j}\lambda_{i,j} \ket{\Phi_{i,j} } \bra{\Phi_{i,j} },     
\end{align}
where $\ket{\Phi_{i,j} }=\sigma_{i,j}(\ket{00}+\ket{11})/\sqrt{2}$, and $\{\lambda_{i,j}\}$ are non-negative coefficients satisfying $\sum_{i,j}\lambda_{i,j}=1$. Note that $\{\sigma_{i,j}\}$ are defined as identity matrix and three Pauli matrices as ones before Eq. \eqref{E:U_ij}, and act on the second qubit. By applying our method, we obtain the lower bound of min-entropy in terms of $\{\lambda_{i,j}\}$, as the following lemma shows.

\begin{lemma}
\label{Lem:min_entropy_rho_AB}
Let Eve holds the purification of $\rho_{AB}$ in Eq. \eqref{rho_AB}, then the composite state shared by Alice, Bob and Eve is given by
\begin{equation}
\ket{\Psi}_{ABE}=\sum_{i,j} \sqrt{\lambda_{i,j}} \ket{\Phi_{i,j}} \otimes \ket{e_{i,j}}.
\end{equation}
Moreover, after performing $A_0$ and obtaining the classical bit $x$, the lower bound of min-entropy is given by 
\begin{equation}
H_{\min}(X|E)_{\rho_{XE}} \ge 1-\log\left(\sqrt{\lambda_{0,0}+\lambda_{1,1}}+\sqrt{\lambda_{1,0}+\lambda_{0,1}} \right)^2,
\end{equation}
where $\rho_{XE}$ is the corresponding CQ state.
\end{lemma}

\begin{proof}
Let $\{\cos\theta\ket{0}+\sin\theta\ket{1}, \sin\theta\ket{0}-\cos\theta\ket{1} \}$ be the eigenvectors of operator $A_0$, then the CQ state after performing $A_0$ and tracing Bob's system is given by 
\begin{align}
 \rho_{XE}={}&\frac{1}{2}\ket{0}\bra{0} \otimes \left[ P\{\cos\theta \frac{\sqrt{\lambda_{0,0}}\ket{e_{0,0}}+\sqrt{\lambda_{1,0}}\ket{e_{1,0}}}{\sqrt{2}} + \sin\theta \frac{\sqrt{\lambda_{0,1}}\ket{e_{0,1}}+\sqrt{\lambda_{1,1}}\ket{e_{1,1}}}{\sqrt{2}}\} \right. \notag \\
& \qquad \left. +P\{\sin\theta \frac{\sqrt{\lambda_{0,0}}\ket{e_{0,0}}-\sqrt{\lambda_{1,0}}\ket{e_{1,0}}}{\sqrt{2}} + \cos\theta \frac{\sqrt{\lambda_{0,1}}\ket{e_{0,1}}-\sqrt{\lambda_{1,1}}\ket{e_{1,1}}}{\sqrt{2}}\}\right] \notag \\
& +\frac{1}{2}\ket{1}\bra{1} \otimes \left[ P\{\cos\theta \frac{\sqrt{\lambda_{0,0}}\ket{e_{0,0}}-\sqrt{\lambda_{1,0}}\ket{e_{1,0}}}{\sqrt{2}} - \sin\theta \frac{\sqrt{\lambda_{0,1}}\ket{e_{0,1}}-\sqrt{\lambda_{1,1}}\ket{e_{1,1}}}{\sqrt{2}}\} \right. \notag \\
& \qquad \left. +P\{\sin\theta \frac{\sqrt{\lambda_{0,0}}\ket{e_{0,0}}+\sqrt{\lambda_{1,0}}\ket{e_{1,0}}}{\sqrt{2}} - \cos\theta \frac{\sqrt{\lambda_{0,1}}\ket{e_{0,1}}+\sqrt{\lambda_{1,1}}\ket{e_{1,1}}}{\sqrt{2}}\}\right].
\end{align}
Without decreasing Eve's power, we assume that Eve holds another ancillary system $F$ with basis $\{\ket{f_0}, \ket{f_1}\}$. Then the corresponding state equals
\begin{align}
 \rho_{XEF}={}&\frac{1}{2}\ket{0}\bra{0} \otimes \left[ P\{\frac{\cos\theta\sqrt{\lambda_{0,0}}\ket{e_{0,0}}+\sin\theta\sqrt{\lambda_{1,1}}\ket{e_{1,1}} + \cos\theta\sqrt{\lambda_{1,0}}\ket{e_{1,0}} +\sin\theta \sqrt{\lambda_{0,1}}\ket{e_{0,1}} } {\sqrt{2}} \otimes \ket{f_0} \} \right. \notag \\
& \qquad \left. +P\{\frac{\sin\theta\sqrt{\lambda_{0,0}}\ket{e_{0,0}} - \cos\theta\sqrt{\lambda_{1,1}}\ket{e_{1,1}} - \sin\theta\sqrt{\lambda_{1,0}}\ket{e_{1,0}} + \cos\theta \sqrt{\lambda_{0,1}}\ket{e_{0,1}} } {\sqrt{2}} \otimes \ket{f_1}\} \right] \notag \\
&+ \frac{1}{2}\ket{1}\bra{1} \otimes \left[ P\{\frac{\cos\theta\sqrt{\lambda_{0,0}}\ket{e_{0,0}}+\sin\theta\sqrt{\lambda_{1,1}}\ket{e_{1,1}} - \cos\theta\sqrt{\lambda_{1,0}}\ket{e_{1,0}} - \sin\theta \sqrt{\lambda_{0,1}}\ket{e_{0,1}} } {\sqrt{2}} \otimes \ket{f_0}\} \right. \notag \\
& \qquad \left. +P\{\frac{\sin\theta\sqrt{\lambda_{0,0}}\ket{e_{0,0}} - \cos\theta\sqrt{\lambda_{1,1}}\ket{e_{1,1}} + \sin\theta\sqrt{\lambda_{1,0}}\ket{e_{1,0}} - \cos\theta \sqrt{\lambda_{0,1}}\ket{e_{0,1}} } {\sqrt{2}} \otimes \ket{f_1}\} \right].
\end{align}
Thus, the sub-normalized state conditioned on $\ket{f_i}$ are respectively
 \begin{subequations}
\begin{align}
 \rho_{XEf_0}={}& \frac{1}{2}\ket{0}\bra{0} \otimes \left[ P\{\frac{\cos\theta\sqrt{\lambda_{0,0}}\ket{e_{0,0}}+\sin\theta\sqrt{\lambda_{1,1}}\ket{e_{1,1}} + \cos\theta\sqrt{\lambda_{1,0}}\ket{e_{1,0}} +\sin\theta \sqrt{\lambda_{0,1}}\ket{e_{0,1}} } {\sqrt{2}} \otimes \ket{f_0} \} \right] \notag \\
 &+\frac{1}{2}\ket{1}\bra{1} \otimes \left[  P\{\frac{\cos\theta\sqrt{\lambda_{0,0}}\ket{e_{0,0}}+\sin\theta\sqrt{\lambda_{1,1}}\ket{e_{1,1}} - \cos\theta\sqrt{\lambda_{1,0}}\ket{e_{1,0}} - \sin\theta \sqrt{\lambda_{0,1}}\ket{e_{0,1}} } {\sqrt{2}} \otimes \ket{f_0}\}  \right],
\end{align}
and
\begin{align}
 \rho_{XEf_1}={}& \frac{1}{2}\ket{0}\bra{0} \otimes \left[ P\{\frac{\sin\theta\sqrt{\lambda_{0,0}}\ket{e_{0,0}}-\cos\theta\sqrt{\lambda_{1,1}}\ket{e_{1,1}} - \sin\theta\sqrt{\lambda_{1,0}}\ket{e_{1,0}} +\cos\theta \sqrt{\lambda_{0,1}}\ket{e_{0,1}} } {\sqrt{2}} \otimes \ket{f_1} \} \right] \notag \\
 &+\frac{1}{2}\ket{1}\bra{1} \otimes \left[  P\{\frac{\sin\theta\sqrt{\lambda_{0,0}}\ket{e_{0,0}}-\cos\theta\sqrt{\lambda_{1,1}}\ket{e_{1,1}} + \sin\theta\sqrt{\lambda_{1,0}}\ket{e_{1,0}} - \cos\theta \sqrt{\lambda_{0,1}}\ket{e_{0,1}} } {\sqrt{2}} \otimes \ket{f_1}\}  \right],
\end{align}
\end{subequations}
Applying Theorem~\ref{Thrm:H_min} to $\rho_{XEf_i}$'s gives
 \begin{subequations}
  \label{E:H_min_X_E_fi}
\begin{equation}
H_{\min}(X|Ef_0)_{\rho_{XEf_0}} = 1- \log\left(\sqrt{\lambda_{0,0}\cos^2\theta+\lambda_{1,1}\sin^2\theta}+\sqrt{\lambda_{1,0}\cos^2\theta+\lambda_{0,1}\sin^2\theta} \right)^2
\end{equation}
and
\begin{equation}
H_{\min}(X|Ef_1)_{\rho_{XEf_1}} = 1- \log\left(\sqrt{\lambda_{0,0}\sin^2\theta+\lambda_{1,1}\cos^2\theta}+\sqrt{\lambda_{1,0}\sin^2\theta+\lambda_{0,1}\cos^2\theta} \right)^2.
\end{equation}
 \end{subequations}
Using the same technique of Eq. \eqref{hminrelation} and Corollary \ref{Cor:H_min_XE}, we obtain the following lower bound of min-entropy of $\rho_{XEF}$, 
\begin{equation}
H_{\min}(X|EF)_{\rho_{XEF}} \ge 1-\log\left(\sqrt{\lambda_{0,0}+\lambda_{1,1}}+\sqrt{\lambda_{1,0}+\lambda_{0,1} } \right)^2 .
 \label{E:H_min_X_EF_bound}
\end{equation}
 Interestingly, this lower bound is \emph{independent} of $\theta$.  Alternatively, we may optimize $\theta$ in Eq.~\eqref{E:H_min_X_E_fi} to obtain a better lower bound for $H_{\min}(X|EF)_{\rho_{XEF}}$.  However, we do not pursue this path as the inequality in Eq.~\eqref{E:H_min_X_EF_bound} is already good enough for our subsequent analysis. Finally, from the data-processing inequality for min-entropy \cite{tomamichel2010duality}, we conclude that
\begin{equation}
\label{H_min_XEF}
H_{\min}(X|E)_{\rho_{XE}} \ge H_{\min}(X|EF)_{\rho_{XEF}}.
\end{equation}
This completes our proof.
\end{proof}

We now come to the generalization of min-entropy calculation for the multi-round case. Following the idea of analysis of single-round case, we reduce the $N$-pair of qubit state $\rho^N_{AB}$ to the generalized Bell-diagonal one. In this variant of DI QKD, we have no reason to assume that the operators used $\{A_0, A_1, B_0, B_1\}$ are the same in each round. In other words, the values of $\alpha$ and $\beta$ in defining these projectors may differ in each round.  Fortunately, in each round of the testing mode, Alice (Bob) has committed to a fixed choice of $\kappa_a \in \{0, 1\}$ ($\kappa_b \in \{0, 1\}$). Therefore, we can make use of the freedom to post-select the reference frame so that each of the operators $\{A_0, A_1, B_0, B_1\}$ in the rounds choosing $\kappa_a=0$ are the same. Here, post-selecting the reference frame refers to re-tagging the reference frame \emph{after completing all rounds}.  As a consequence, $\rho^N_{AB}$ can be viewed as a permutation invariant state \cite{renner2008security,renner2007symmetry} since applying any permutation operator on $\rho^N_{AB}$ would not change the observable statistics and would not decrease Eve's power. Now, let us consider a virtual case related to the actual protocol in which Alice and Bob choose the testing mode only, and the number of testing round is $4(n+k)$ with each combination $(\kappa_a, \kappa_b) \in \{00, 01, 10, 11\}$ by $n+k$. In this virtual protocol, their measurement is 
\begin{equation}
\label{C}
C^N:=(A_0 \otimes B_0)^{\otimes n+k} \otimes (A_0 \otimes B_1)^{\otimes n+k} \otimes (A_1 \otimes B_0)^{\otimes n+k} \otimes (-A_1 \otimes B_1)^{\otimes n+k},
\end{equation}
and is followed by a uniformly random permutation operation $\pi$ selected from $\mathcal{S}_N$, the set of permutations on $\{1, \cdots, N\}$.  Note that the eigenstates of $S$ are the four Bell-states $\{\ket{\Phi_{i,j}}\}$.  Hence, applying Bell-measurement to $\rho^N_{AB}$ would not change its observable statistics. In the asymptotic case that the number of rounds $N \to \infty$, the outcome distributions of applying $C^N$ and $S^{\otimes N}$ are identical. But when $N$ is finite, these two statistical distributions may differ.  Fortunately, we could bound the observable statistics of $C^N$ by bounding the one of $S^{\otimes N}$ as well as the statistical differences between the outcomes of $S^{\otimes N}$ and $C^N$ using Chernoff-Hoeffding inequality \cite{hoeffding1994probability}.

\begin{lemma}
\label{Lem:error_between_two_operators}
Let $\rho^N$ be a permutation invariant $N$-partite state with $N$ being an even number. Let $\{\mathcal{F}, \mathcal{F}_0, \mathcal{F}_1\}$ be measurement operators acting on a single system of $\rho^N$, with $\mathcal{F}=(\mathcal{F}_0+\mathcal{F}_1)/2$. Suppose that $\{\mathcal{F}, \mathcal{F}_0, \mathcal{F}_1\}$ have two-output classical bit $\{0,1\}$. Denote the output $N$-bit string after applying $\mathcal{F}_0^{\otimes N/2} \otimes \mathcal{F}_1^{\otimes N/2}$ to $\rho^N$ by $l$. Denote the frequency of $0$ by $l_0$. Let the frequency of $0$ after applying $\mathcal{F}^{\otimes N}$ to $\rho^N$ be $f_0$, then
\begin{equation}
\Pr \left[|f_0-l_0| \ge \mu \right] \le 2e^{-2\mu^2N}
\end{equation}
for all $\mu > 0$.
\end{lemma}

\begin{proof}
 Clearly, applying $\mathcal{F}$ to a single system can be viewed as applying $\mathcal{F}_0$ and $\mathcal{F}_1$ uniformly randomly. Let $\mathcal{X} \in \{0, 1\}$ be an \iid random variable with mean value $1/2$. If $\mathcal{X}=0 \ (1)$, $\mathcal{F}_0$ ($\mathcal{F}_1$) is applied to the system. Then, for any $\mu \ge 0$, Chernoff-Hoeffding inequality \cite{hoeffding1994probability} implies that
\begin{equation}
\Pr\left[|\mathbb{E}(\mathcal{X})-\frac{1}{2}| \ge \mu\right] \le 2e^{-2\mu^2N}
\end{equation} 
where $\mathbb{E}(\mathcal{X})=\sum_{i=1}^N \mathcal{X}_i/N$. For any output bit-string $l$, half of bits are resulted from $\mathcal{F}_0$ while the another half of bits are resulted from $\mathcal{F}_1$. Assuming that each bit of $l$ will flip if its corresponding measurement operator changes, then applying $\mathcal{F}^{\otimes N}$ will result in at most $N\mu$ bits changed with small failure probability of $2e^{-2\mu^2N}$. This completes the proof.
\end{proof}

\begin{corollary} 
\label{Cor:the_difference_between_S^N_and_C^N}
Let $\mathcal{C}^{N}$ and $\mathcal{S}^{\otimes N}$ be the corresponding measurement maps of $C^{N}$ and $S^{\otimes N}$, respectively. Given an output of $\mathcal{C}^{N}$ with CHSH game's winning frequency $\omega$ (the frequency of $0$), then the winning frequency of $\mathcal{S}^{\otimes N}$, denoted as $\omega_S$, satisfies
\begin{equation}
\label{omega_S-omega}
\Pr \left[|\omega_S-\omega| \ge \mu+\nu \right] \le 2e^{-2\mu^2N}+4e^{-\nu^2N}
\end{equation}
for any $\mu, \nu > 0$.
\end{corollary}

\begin{proof}
Let $H^N$ be an intermediate measurement operator, defined as
\begin{equation}
H^N:=\left(\frac{A_0+A_1}{2} \otimes B_0 \right)^{\otimes 2(n+k)} \otimes \left(\frac{A_0-A_1}{2} \otimes B_1\right)^{\otimes 2(n+k)}.
\end{equation}
By the fact that 
\begin{equation}
S=\frac{1}{2}\left(\frac{A_0 + A_1}{2} \otimes B_0 +\frac{ A_0 - A_1}{2} \otimes B_1\right),
\end{equation}
 and Lemma \ref{Lem:error_between_two_operators}, the error when respectively applying $H^N$ and $S^{\otimes N}$ on $\rho^N_{AB}$ is bounded by $\mu$ with small failure probability of $2e^{-2\mu^2N}$. Fixing $B_0$ on Bod's side and using Lemma \ref{Lem:error_between_two_operators} again, the error when respectively applying $A_0^{\otimes n+k} \otimes A_1^{\otimes n+k}$ and $(\frac{A_0 + A_1}{2})^{\otimes 2(n+k)}$ on $\rho^N_{AB}$ is bounded by $\nu$ with small failure probability of $2e^{-\nu^2N}$. Things are similar if fixing $B_1$ on Bob's side. Therefore, the total failure probability is bounded by
\begin{equation}
1-(1-2e^{-2\mu^2N})(1-2e^{-\nu^2N})(1-2e^{-\nu^2N}) \le 2e^{-2\mu^2N}+4e^{-\nu^2N}.
\end{equation}
This completes the proof.
\end{proof}

Then, we come to the analysis of statistical fluctuation. Owing to Eq. \eqref{omega_S-omega}, we have 
\begin{equation}
\Pr \left[|\omega_S-\omega| \ge \frac{\mu}{2}+\frac{\mu}{2} \right] \le 2e^{-2\mu^2k}+4e^{-\mu^2k} < 6 e^{-\mu^2k}
\end{equation}
for all $\mu, k > 0$.
Let $\varepsilon^2_t=6 e^{-\mu^2k}$, then $\mu=\sqrt{\frac{2}{k}\ln \frac{\sqrt{6}}{\varepsilon_t}}$. Thus, the lower bound $\omega_S$ for testing mode can be written as a function of parameters $\omega$, $k$ and the small failure probability $\varepsilon^2_t$, namely,
\begin{equation}
\omega_S=\omega-\sqrt{\frac{2}{k}\ln \frac{\sqrt{6}}{\varepsilon_t}}.
\end{equation}
Picking Serfling inequality \cite{tomamichel2012tight}, the lower bound $\hat{\omega}_S$ for key generation mode can be written as a function of parameters $\omega$, $n$, $k$ and the small failure probability $\varepsilon^2_g$, namely,
\begin{equation}
\hat{\omega}_S=\omega_S-\sqrt{\frac{(n+k)(4k+1)}{16nk^2}\ln \frac{1}{\varepsilon_g}}.
 \label{E:omega}
\end{equation}

Up to now, we have reduced $\rho^N_{AB}$ to the generalized Bell-diagonal state, at the cost of lower winning frequency. The state regarding to key generation round, denoted as $\rho^{4n}_{AB}$, is also Bell-diagonal. Similar to the case of single round, we write $\rho^{4n}_{AB}$ as
\begin{equation}
\label{rhoab4n}
\rho^{4n}_{AB}:=\sum_{\bm{i}, \bm{j}} \lambda_{\bm{i}, \bm{j}} \ket{\Phi_{\bm{i}, \bm{j}} } \bra{\Phi_{\bm{i}, \bm{j}} },
\end{equation}
which differs from $\rho^{n}_{AB}$ in Eq. \eqref{rhoabn} by the number of subsystem. In the following corollary, we use $\{\lambda_{\bm{i}, \bm{j}}\}$ to express the lower bound of min-entropy.

\begin{corollary}
\label{Cor:min_entropy_rho^4n_AB}
Let Eve holds the purification of $\rho^{4n}_{AB}$ in Eq. \eqref{rhoab4n}, then the composite state shared by Alice, Bob and Eve is given by
\begin{equation}
\ket{\Psi}^{4n}_{ABE}=\sum_{\bm{i},\bm{j}} \sqrt{\lambda_{\bm{i},\bm{j}}} \ket{\Phi_{\bm{i},\bm{j}}} \otimes \ket{e_{\bm{i},\bm{j}}}.
\end{equation}
After performing $A^{\otimes 4n}_0$ and obtaining classical bit string $\bm{x}$, the lower bound of min-entropy is given by 
\begin{equation}
H_{\min}(\bm{X}|E) \ge n-\log\left(\sqrt{\sum_{\bm{h}}\lambda_{\bm{h}}} \right)^2,
\end{equation}
where $\bm{h}:=\bm{i}+\bm{j}$ and $\bm{h} \in \{0,1\}^{4n}$.
\end{corollary}

\begin{proof}
 This can be proven using the same logical argument in the proof of
Lemma~\ref{Lem:min_entropy_rho_AB}.
\end{proof}

In the following lemma, we connect $\{\lambda_{\bm{h}}\}$ with $\omega$, and thus express the smoothed version of $H_{\min}(\bm{X}|E)$ in terms of $\omega$. 

\begin{lemma} 
\label{Lem:smooth_H_min_DI_QKD}
Let $\varepsilon=\varepsilon_t+\varepsilon_g \ge 0$. Define the set
\begin{equation}
\mathcal{S}_{\hat{\omega}_S}:=\left\{\bm{h} \colon \frac{\sum_k \bm{h}_k }{4n} \le \frac{1-\sqrt{16\hat{\omega}_S(\hat{\omega}_S-1)+3}}{2} \right\} .
\end{equation}
Suppose ${\sum_{\bm{h} \in \mathcal{S}_{\hat{\omega}_S} } \lambda_{\bm{h}}}=1-\varepsilon$. Then
\begin{equation}
\label{H_min^varepsilon_4n}
H_{\min}^{\varepsilon}(\bm{X}|E) \ge 4n \left[1-h(\frac{1- \sqrt{16\hat{\omega}_S(\hat{\omega}_S-1)+3}}{2}) \right].
\end{equation}
\end{lemma}

\begin{proof}
Recalling the CHSH operator in Eq. \eqref{CHSH}, we rewrite it as 
\begin{align}
 S ={}& \frac{1}{2} \left[\sqrt{\Lambda_+} \left(P\{\ket{\Phi_{0,0}}\}+P\{\ket{\Phi_{1,0}}\} - P\{\ket{\Phi_{0,1}}\} - P\{\ket{\Phi_{1,1}}\} \right) \right. \notag \\
  & \left. + \sqrt{\Lambda_-}  \left(P\{\ket{\Phi_{0,0}}\}-P\{\ket{\Phi_{1,0}}\} + P\{\ket{\Phi_{0,1}}\} - P\{\ket{\Phi_{1,1}}\} \right)\right].
\end{align}  
This CHSH operator corresponds to a Bell-measurement followed by re-tagging these measurement results. Concretely, denote the frequencies of outputs $\ket{\Phi_{0,0}}$, $\ket{\Phi_{1,0}}$, $\ket{\Phi_{0,1}}$ and $\ket{\Phi_{1,1}}$ by $f_{00}$, $f_{10}$, $f_{01}$ and $f_{11}$, respectively. Clearly, $f_{00}+f_{10}+f_{01}+f_{11}=1$, and the winning frequency is given by 
\begin{align}
\omega_S &= \frac{1}{4}  \left[2 + \sqrt{\Lambda_+} \left(f_{00} - f_{11} +  f_{10} - f_{01} \right) + \sqrt{\Lambda_-}\left( f_{00} - f_{11} - f_{10} + f_{01}\right)    \right] \notag \\
                 & \le \frac{1}{4} \left[2 + \sqrt{2 (f_{00} - f_{11})^2 + 2( f_{10} - f_{01})^2} \right],
\end{align}  
where we make use of Cauchy-Schwarz inequality and the fact that $\Lambda_++\Lambda_-=1$.
Thus, 
\begin{align}
\sqrt{16 \omega_S(\omega_S-1)+3} &= \sqrt{(4\omega_S-2)^2 - 1} \notag\\
& \le \sqrt{ 2 (f_{00} - f_{11})^2 + 2( f_{10} - f_{01})^2- 1} \notag\\
& \le \sqrt{ 2 (f_{00} + f_{11})^2 + 2( f_{10} + f_{01})^2- 1} \notag\\
& =1-2( f_{10} + f_{01}),
\end{align}  
where we make use of fact that $f_{00}+f_{10}+f_{01}+f_{11}=1$ in the last equality. Then, 
\begin{equation}
f_{10} + f_{01} \le \frac{1- \sqrt{16\omega_S(\omega_S-1)+3}}{2}. 
\end{equation}
Then we can prove this lemma following the same logical argument in the proof of Lemma \ref{Lem:smooth_H_min}.
\end{proof}

Finally, by applying the quantum leftover hashing lemma \cite{renner2008security,tomamichel2011leftover}, we obtain a lower bound for the secret key length $\ell$ and prove that the protocol is $\varepsilon_\text{sec}$-secret with $\varepsilon_\text{sec}=4\varepsilon$. 

\begin{theorem} 
\label{Thrm:DIQKDkeylength}
If the final key length $\ell$ obeys
\begin{equation}
\label{key_length2}
 \ell \le 4n \left[1-h(\frac{1- \sqrt{16\hat{\omega}_S(\hat{\omega}_S-1)+3}}{2}) \right]-{\normalfont\text{leak}_{\text{EC}}}- \log\frac{2}{\varepsilon_{\normalfont\text{sec}}^2\varepsilon_{\normalfont\text{cor}}}, 
\end{equation}
 where the classical information of the error correction leaked to Eve is at most ${\normalfont\text{leak}_{\text{EC}}}+\log_{2}(1/\varepsilon_{\normalfont\text{cor}})$, then this protocol is $\varepsilon_{\normalfont\text{sec}}$-secret. 
\end{theorem}

\begin{proof}
 This can be proven using the same logical argument in the proof of
Theorem~\ref{Thrm:keylength}.
\end{proof}

By incorporating our min-entropy calculation method with standard secret key rate computation techniques, we obtain a tight key rate which deviates from the asymptotic result only by terms that are caused by unavoidable statistical fluctuations in the parameter estimation step. In contrast, Ref.~\cite{masanes2011secure} only provides a loose bound if noisy channel taken into consideration whereas  Ref.~\cite{lim2013device} derives a key rate from entropic uncertain relations that is looser than ours.  Finally, we remark that the security proof in Ref.~\cite{lim2013device} could be adapted to our min-entropy approach even though their physical implementation is different from ours.

\subsection{Application To Continuous Variable Quantum Random Number Generation}

We now report a new SI-CV QRNG protocol and prove its security using our one-shot min-entropy calculation approach. The idea of this protocol is inspired by the number–phase uncertainty relation \cite{gerry2023introductory} of electromagnetic field.   Concretely, the randomness stems from the fact that the more certain the photon number is, the more uncertain the phase will be. Therefore, if we ensure that the standard deviation of the photon number of an incoming light is sufficiently small, then the phase of this light must be close to a uniform \iid distribution. In this way, we do not need to trust the incoming light as long as we could test both of phase and photon number. The problem of this approach is that it is not clear how to define a general quantum phase operator (see Sec. 2.7 in Ref. \cite{gerry2023introductory} for discussion).  Fortunately, we may substitute the quantum phase operator by heterodyne detection.  More precisely, the ``phase'' can be determined by the ratio of the two quadratures of an heterodyne detection.  As for photon number testing, the most direct way is to use a photon number resolving detector.  Nonetheless, this kind of detector is impractical due to low count rate and high cost. Here, we use a more common setup by using a threshold detector to test how close the incoming light is to the vacuum state.  We write down the procedure of our SI-CV QRNG protocol in Procedure~\ref{procedure_QRNG}.

\begin{table}
 \caption{\label{procedure_QRNG}Source-Independent Quantum Random Number Generation Protocol}
 \begin{enumerate}
  \item \emph{Sending untrusted states:} Eve prepares an $N$-partite optical quantum state and sends them to Alice one by one.
  \item \emph{Measurement:} For each photon send by Eve, Alice randomly chooses either the randomness generation mode or the testing mode.  Denote the number of photons used in the randomness generation mode and testing mode by $n$ and $k$, respectively. Clearly, $n+k=N$.  We fix $k$ so that $N \gg k$. If randomness generation mode chosen, Alice performs an heterodyne measurement on the received optical pulse to obtain the phase $\theta \in [0,2\pi)$ and amplitude $\mu \in [0,+\infty)$. If testing mode chosen, Alice performs a single photon measurement on the received optical pulse.  She records the frequency of detection $Q$, namely, the number of detection events divided by $k$.
  \item \emph{Parameter estimation:} Alice continues the protocol only if $Q$ is smaller the predefined threshold.
  \item \emph{Discretization:} Alice maps the continuous number $\theta \in [0, 2\pi)$ to a discrete number $x \in \{0, 1, 2, 3\}$. Specifically, $x=0$ when $\theta \in [0, \pi/2)$, $x=1$ when $\theta \in [\pi/2, \pi)$, $x=2$ when $\theta \in [\pi, 3\pi/2)$, and $x=3$ when $\theta \in [3\pi/2, 2\pi)$.  In this way, the sequence of continuously distributed $\theta$'s is mapped to the raw sequence of discrete random variables.
  \item \emph{Randomness extraction:} Alice applies a random two-universal hash function to the raw sequence of $x$ to extract final secret $\ell$-bit random numbers.
\end{enumerate}
\end{table}

Two remarks are in order.  First, we discretize the phase into four regions in the above protocol for illustrative purpose.  It is perfectly fine to sub-divide the phase into any equally spaced regions.  The security analysis is essentially unchanged.
Second, the energy test can be accomplished by heterodyne detection \cite{renner2009finetti,upadhyaya2021dimension,kanitschar2023finite}.  Thus, it seems possible to execute a similar protocol without threshold detector.  However, we do not pursue this investigation here as it is beyond the main goal of this Supplemental Material.

We begin our security analysis by applying our one-shot approach to compute the min-entropy of a single round of randomness generation mode.  An heterodyne measurement corresponds to the following POVM
\begin{equation}
\Pi_{\mu,\theta}:=\frac{1}{\pi}\ket{\sqrt{\mu}e^{i\theta}}\bra{\sqrt{\mu}e^{i\theta}},
\end{equation}
where $\mu \in [0, \infty)$ and $\theta \in [0, 2\pi)$. The outputs of an heterodyne measurement are two quadratures denoted by $\{q, p\}$.  They are related to $\mu$ and $\theta$ by $\mu=q^2+p^2$ and $\theta=\arctan (q / p)$. A threshold measurement corresponds to the POVM with two elements $\{ \ket{0}\bra{0},  \mathit{I}-\ket{0}\bra{0} \}$, where $\ket{0}\bra{0}$ and $\mathit{I}-\ket{0}\bra{0}$ correspond to ``no click event'' (that is, non-detection) and ``click event'' (that is, detection), respectively.  The untrusted light sent to Alice can be described by an arbitrary energy-bounded optical quantum state $\varrho_A$. However, we can reduce it to a diagonalized one in the Fock state basis. This is because Eve may apply a phase-randomized operation to $\varrho_A$ and record the corresponding phase information.
This phase-randomized operation would not change the detection frequency $Q$.  Therefore, we write $\varrho_A$ as
\begin{equation}
 \varrho_A = \sum_{m=0}^{\infty} p_m\ket{m}\bra{m}
\end{equation}
with constraints of normalized condition
\begin{subequations}
\begin{equation}
 \label{normalized_cond}
 \sum_{m=0}^{\infty} p_m=1
\end{equation}
 and energy-bounded condition
\begin{equation}
 \label{energy-bounded_cond}
  L:=\sum_{m=0}^{\infty} mp_m < \infty .
\end{equation}
\end{subequations}
Since Eve holds the purification of $\varrho_A$, the composite state can be written as $\ket{\Phi}_{AE}=\sum_{m=0}^{\infty} \sqrt{p_m}\ket{m}\ket{e_m}$, where $ \{\ket{e_m}\}$ are hold by Eve. For an implementation in practice, any heterodyne measurement is discretized. As stated in discretization step of our SI-CV QRNG protocol, we use $x$ instead of $\theta$ to generate the final random number sequence.  The corresponding POVM elements for $x$ are given by
\begin{equation}
\Pi_x:= \frac{1}{\pi} \int _{\frac{x\pi}{2}}^{\frac{(x+1)\pi}{2}} d\theta \int _{0}^{\infty} \sqrt{\mu} d\sqrt{\mu} \ \Pi_{\mu,\theta} =\frac{1}{2\pi} \int _{\frac{x\pi}{2}}^{\frac{(x+1)\pi}{2}} d\theta \int _{0}^{\infty} d\mu \ \Pi_{\mu,\theta}.
\end{equation}
Therefore, the CQ state after applying POVM $\{ \Pi_x \}$ to $\ket{\Phi}_{AE}$ equals
\begin{align}
\label{varrho_XE}
 \varrho_{XE} :={}& \sum_{x=0}^{3} \ket{x}\bra{x} \otimes \Tr(\Pi_x \ket{\Phi}_{AE}\bra{\Phi}) \nonumber \\
		      ={}&  \sum_{x=0}^{3} \ket{x}\bra{x} \otimes \frac{1}{2\pi} \int _{\frac{x\pi}{2}}^{\frac{(x+1)\pi}{2}} d\theta \int _{0}^{\infty} d\mu \ P\{ \bra{\sqrt{\mu}e^{i\theta}}\sum_{m=0}^{\infty} \sqrt{p_m}\ket{m}\ket{e_m} \} \nonumber \\
		     ={}&  \sum_{x=0}^{3} \ket{x}\bra{x} \otimes \frac{1}{2\pi} \int _{\frac{x\pi}{2}}^{\frac{(x+1)\pi}{2}} d\theta \int _{0}^{\infty} d\mu \ P\{ \sum_{m=0}^{\infty} \sqrt{\frac{e^{-\mu}\mu^mp_m}{m!}} \ e^{-im\theta}\ket{e_m} \},
\end{align}
where we make use of the fact that 
\begin{equation}
\ket{\sqrt{\mu}e^{i\theta}}=\sum_{m=0}^{\infty}\sqrt{\frac{e^{-\mu}\mu^m}{m!}}\ e^{im\theta}\ket{m}.
\end{equation}
With the replacement of variables, $\theta \to x\pi/2+\theta$ and $m \to 4m+y$ where $y \in \{0, 1, 2, 3\}$, we rewrite $\varrho_{XE}$ in Eq. \eqref{varrho_XE} as 
\begin{align}
 \varrho_{XE} ={}&\sum_{x=0}^{3} \ket{x}\bra{x} \otimes \frac{1}{2\pi} \int _{0}^{\frac{\pi}{2}} d\theta \int _{0}^{\infty} d\mu \ P\{ \sum_{y=0}^{3}\sum_{m=0}^{\infty} \sqrt{\frac{e^{-\mu}\mu^{4m+y}p_{4m+y}}{(4m+y)!}} \ e^{-i(4m+y)(\frac{x\pi}{2}+\theta)}\ket{e_{4m+y}} \} \nonumber \\
		     ={}&\frac{1}{4}\sum_{x=0}^{3} \ket{x}\bra{x} \otimes \frac{2}{\pi} \int _{0}^{\frac{\pi}{2}} d\theta \int _{0}^{\infty} d\mu \ P\{ \sum_{y=0}^{3}e^{-i\frac{\pi}{2}xy} \sum_{m=0}^{\infty} \sqrt{\frac{e^{-\mu}\mu^{4m+y}p_{4m+y}}{(4m+y)!}} \ e^{-i(4m+y)\theta}\ket{e_{4m+y}} \} \nonumber \\
		     ={}&\frac{2}{\pi} \int _{0}^{\frac{\pi}{2}} d\theta \int _{0}^{\infty} d\mu \ \frac{1}{4}\sum_{x=0}^{3} \ket{x}\bra{x} \otimes P\{ \sum_{y=0}^{3} e^{-i\frac{\pi}{2}xy} \ket{\bar{e}_y^{\theta,\mu}} \},
\end{align}
where we define the sub-normalized states $\{\ket{\bar{e}_y^{\theta,\mu}} \}$ as
\begin{equation}
\label{bar_e_y}
 \ket{\bar{e}_y^{\theta,\mu}}:=\sum_{m=0}^{\infty} \sqrt{\frac{e^{-\mu}\mu^{4m+y}p_{4m+y}}{(4m+y)!}} \ e^{-i(4m+y)\theta}\ket{e_{4m+y}}.
\end{equation}
Note that $\ket{\bar{e}_y^{\theta,\mu}}$ is a sub-normalized state owing to the energy-bounded condition in Eq.~\eqref{energy-bounded_cond}.
Let us define the ``joint probability density'' of variables $\{ \theta, \mu \}$ by
\begin{equation}
\label{varrho_XE_theta_mu}
 \varrho_{XE}^{\theta, \mu}:=\frac{1}{4}\sum_{x=0}^{3} \ket{x}\bra{x} \otimes P\{ \sum_{y=0}^{3}e^{-i\frac{\pi}{2}xy} \ket{\bar{e}_y^{\theta,\mu}} \},
\end{equation}
whose form matches that of Eq. \eqref{rho_XE}.  Clearly,
\begin{equation}
\label{int_varrho_XE}
\varrho_{XE} =\frac{2}{\pi} \int _{0}^{\frac{\pi}{2}} d\theta \int _{0}^{\infty} d\mu \ \varrho_{XE}^{\theta, \mu}.
\end{equation}

Now, we are ready to calculate the min-entropy lower bound of $\varrho_{XE}$ in Eq. \eqref{varrho_XE}. 

\begin{lemma}
\label{Lem:varrho_XE}
Let 
\begin{equation}
\label{varrho^theta_XE}
\varrho_{XE}^{\theta}:=\int _{0}^{\infty} d\mu \ \varrho_{XE}^{\theta, \mu},
\end{equation}
then
\begin{equation}
H_{\min}(X|E)_{\varrho_{XE}} \ge H_{\min}(X|E)_{\varrho^{\theta}_{XE}},
\end{equation}
for any $\theta \in [0, \pi/2)$.
\end{lemma}

\par \medskip

\begin{proof}
Recall the form of $\varrho_{XE}^{\theta, \mu}$ in Eq. \eqref{varrho_XE_theta_mu}, we will show that $\varrho_{XE}^{\theta, \mu}$ is equivalent to $\varrho_{XE}^{\theta', \mu}$ under unitary transformation for any $\theta, \theta' \in [0, \pi/2]$. Define the phase-shifting like operator as 
\begin{equation}
U(\delta)=\sum_{m=0}^{\infty} e^{-im\delta} \ket{e_m}\bra{e_m},
\end{equation}
where $\delta$ is the rotating angle. Take $\delta=\theta'-\theta$, then
\begin{equation}
\Big(\mathit{I}_X \otimes U(\delta)\Big) \varrho_{XE}^{\theta, \mu} \Big(\mathit{I}_X \otimes U(\delta)\Big)^{\dagger}=\varrho_{XE}^{\theta', \mu}.
\end{equation}
Hence, we obtain that $H_{\min}(X|E)_{\varrho_{XE}^{\theta, \mu}} = H_{\min}(X|E)_{\varrho_{XE}^{\theta', \mu}}$, and similarly 
\begin{equation}
\label{H_min_varrho_XE}
H_{\min}(X|E)_{\varrho_{XE}^{\theta}} = H_{\min}(X|E)_{\varrho_{XE}^{\theta'}}.
\end{equation} 
Next, we focus on the term of $\bra{e_m}\bra{x}\varrho_{XE}\ket{x}\ket{e_{m'}}$ in Eq. \eqref{varrho_XE}, which equals to 
\begin{equation}
\int _{\frac{x\pi}{2}}^{\frac{(x+1)\pi}{2}}  \frac{d\theta}{2\pi} \ e^{-i(m-m')\theta} \int _{0}^{\infty} d\mu \ e^{-\mu}\sqrt{\frac{\mu^m\mu^{m'} p_m p_{m'}}{m! m'!}}=\Delta_{mm'} \frac{e^{-i(m-m')(\frac{x\pi}{2}+\frac{\pi}{4})}}{4} \int _{0}^{\infty} d\mu \ e^{-\mu}\sqrt{\frac{\mu^m\mu^{m'} p_m p_{m'}}{m! m'!}},
\end{equation}
where 
\begin{equation}
 \Delta_{mm'}= \begin{cases}
  \frac{\sin \frac{(m-m')\pi}{4}}{\frac{(m-m')\pi}{4}} & \text{if~} m \ne m' , \\
  1 & \text{if~} m = m' .
 \end{cases}
\end{equation}
Note that $0 \le \Delta_{mm'} \ \textless \ 1$ if $m \neq m'$. Moreover, by setting $\theta=\pi/4$ and combining the definition of $\varrho_{XE}^{\theta}$ in Eq. \eqref{varrho^theta_XE}, we find 
\begin{equation}
\bra{e_m}\bra{x}\varrho_{XE}\ket{x}\ket{e_{m'}} = \Delta_{mm'} \bra{e_m}\bra{x}\varrho^{\pi/4}_{XE}\ket{x}\ket{e_{m'}}.
\end{equation}
In this view, we further define the phase flip like POVM within the two-dimensional sub-space spanned by $\{\ket{e_m}, \ket{e_{m'}} \}$, which is given by 
\begin{equation}
\mathcal{E}_{mm'} (\rho)= \sum_{\ell, \ell'=m, m'} \Delta_{\ell\ell'} \ket{e_\ell}\bra{e_\ell} \rho \ket{e_{\ell'}}\bra{e_{\ell'}},
\end{equation}
for any $\rho$. Composing all possible channels $\mathcal{E}_{mm'}$ with constraint $m < m'$, we obtain the new POVM, given by
\begin{equation}
\mathcal{E} = \bigcircle_{m < m'} \mathcal{E}_{mm'} \equiv \mathcal{E}_{01} \circ \mathcal{E}_{02} \circ \mathcal{E}_{03} \circ \cdots \circ \mathcal{E}_{12} \circ \mathcal{E}_{13} \circ \mathcal{E}_{14} \circ \cdots \circ \mathcal{E}_{23} \circ \mathcal{E}_{24} \circ \mathcal{E}_{25} \circ \cdots .
\end{equation}
(Note that channel in the above composition commutes with each other so the order of composition is immaterial.)
We immediately have
\begin{equation}
(\mathit{I}_X \otimes \mathcal{E}) (\varrho^{\pi/4}_{XE})=\varrho_{XE}.
\end{equation}
According to the data processing inequality for states within infinite but separable Hilbert spaces \cite{furrer2011min}, we arrive at 
\begin{equation}
H_{\min}(X|E)_{\varrho_{XE}} \ge H_{\min}(X|E)_{\varrho^{\pi/4}_{XE}}.
\end{equation}
By Eq. \eqref{H_min_varrho_XE}, we complete the proof.
\end{proof}

\begin{corollary} 
\label{Cor:lower_bound_Hmin_varrho_XE}
Let 
\begin{align}
\label{bar_varrho_XE_theta_mu}
 \bar{\varrho}_{XE}:=& \frac{1}{4}\sum_{x=0}^{3} \ket{x}\bra{x} \otimes P\{ \sum_{m=0}^{\infty} \sqrt{p_m} \ e^{-im\frac{x\pi}{2}}\ket{e_m} \} \\
                               =& \frac{1}{4}\sum_{x=0}^{3} \ket{x}\bra{x} \otimes P\{ \sum_{y=0}^{3}e^{-i\frac{\pi}{2}xy} \sum_{m=0}^{\infty} \sqrt{p_{4m+y}}  \ket{e_{4m+y}} \}     
\end{align}
then 
\begin{equation}
 \label{H_min_RNG}
H_{\min}(X|E)_{\varrho_{XE}} \ge H_{\min}(X|E)_{ \bar{\varrho}_{XE}},
\end{equation}
and thus
\begin{equation}
 \label{H_min_RNG}
H_{\min}(X|E)_{\varrho_{XE}} \ge 2- \log \left(\sum_{y=0}^{3}\sqrt{\sum_{m=0}^{\infty} p_{4m+y}} \right)^2.
\end{equation}
\end{corollary}

\begin{proof}
From Eqs.  \eqref{varrho_XE}, \eqref{bar_e_y}, \eqref{varrho_XE_theta_mu} and \eqref{varrho^theta_XE}, we obtain
\begin{equation}
\varrho_{XE}^{0}=  \frac{1}{4}\sum_{x=0}^{3} \ket{x}\bra{x} \otimes \int _{0}^{\infty} d\mu \ P\{ \sum_{m=0}^{\infty} \sqrt{\frac{e^{-\mu}\mu^mp_m}{m!}} \ e^{-im\frac{x\pi}{2}}\ket{e_m} \}.
\end{equation}
Thus,
\begin{equation}
\bra{e_m}\bra{x}\varrho^0_{XE}\ket{x}\ket{e_{m'}} =
\frac{e^{-i(m-m')\frac{x\pi}{2}}}{4} \sqrt{p_m p_{m'}} \int _{0}^{\infty} d\mu \ e^{-\mu}\sqrt{\frac{\mu^m\mu^{m'}}{m! m'!}}.
\end{equation}
Then, we define that the non-negative coefficients
\begin{equation}
\Omega_{mm'}:=\int _{0}^{\infty} d\mu \ e^{-\mu}\sqrt{\frac{\mu^m\mu^{m'}}{m! m'!}} \le \sqrt {\int _{0}^{\infty} d\mu \ e^{-\mu} \frac{\mu^m}{m!} \ \int _{0}^{\infty} d\mu \ e^{-\mu} \frac{\mu^{m'}} {m'!} } =1.
\end{equation}
Here we make use of the Cauchy–Schwarz inequality to obtain this inequality, and clearly, the equality holds if and only if $m=m'$. Hence, $\bra{e_m}\bra{x}\varrho^0_{XE}\ket{x}\ket{e_{m'}}=\Omega_{mm'}\bra{e_m}\bra{x}\bar{\varrho}_{XE}\ket{x}\ket{e_{m'}}$. Following the same logic of Lemma~\ref{Lem:varrho_XE}, there exists a POVM $\mathcal{F}$ such that 
\begin{equation}
(\mathit{I}_X \otimes \mathcal{F}) (\bar{\varrho}_{XE})=\varrho^0_{XE}.
\end{equation}
By using data processing inequality again and the result of Lemma~\ref{Lem:varrho_XE}, we have
\begin{equation}
H_{\min}(X|E)_{\varrho_{XE}} \ge H_{\min}(X|E)_{\varrho^0_{XE}} \ge H_{\min}(X|E)_{ \bar{\varrho}_{XE}}.
\end{equation}
Eq. \eqref{H_min_RNG} follows directly from Theorem~\ref{Thrm:H_min} and this completes our proof.
\end{proof}

We now come to the generalization of min-entropy calculation for the multi-round case. Following the idea of analysis of single-round case, we reduce the incoming $n$-partite state denoted by $\varrho_A^n$ to a diagonalized one in the Fock state basis so that $\varrho_A^n$ is written as 
\begin{equation}
\varrho_A^n=\sum_{\bm{m}} p_{\bm{m}} \ket{\bm{m}}\bra{\bm{m}},
\end{equation}
where $\bm{m} \in \mathbb{R}^n$, the sum is over all possible $\bm{m}$ and 
\begin{equation}
\ket{\bm{m}}:= \bigotimes_{k=1}^{n} \ket{m_k}, 
\end{equation}
$m_k$ is the $k$-th dit of string $\bm{m}$, and $\{\ket{m_k}\}$ are the Fock states for the $k$-th sub-system. Since Eve holds the purification of $\varrho_A^n$, the composite state is written by 
\begin{equation}
\label{Phi_n_AE}
\ket{\Phi}^n_{AE}=\sum_{\bm{m}} \sqrt{p_{\bm{m}}} \ket{\bm{m}}\ket{e_{\bm{m}}},
\end{equation}
where $\{ \ket{e_{\bm{m}}} \}$ are hold by Eve. Then, the CQ state, after applying POVM $\{ \Pi_x \}$ to each sub-system of $\ket{\Phi^n}_{AE}$, is given by
\begin{align}
 \varrho_{\bm{X}E} :={}& \sum_{\bm{x}} \ket{\bm{x}}\bra{\bm{x}} \otimes \frac{1}{(2\pi)^n} \int _{\frac{x_1\pi}{2}}^{\frac{(x_1+1)\pi}{2}} d\theta_1 \cdots \int _{\frac{x_n\pi}{2}}^{\frac{(x_n+1)\pi}{2}} d\theta_n  \int _{0}^{\infty} d\mu_1 \cdots  \int _{0}^{\infty} d\mu_n \ \times \nonumber  \\
& \qquad P\{ \sum_{\bm{m}}\sqrt{ p_{\bm{m}} \prod_{k=1}^n \frac{e^{-\mu_k}\mu^{m_k}}{m_k!}} e^{-i\bm{m} \cdot \bm{\theta}}\ket{e_{\bm{m}}} \},
\end{align}
where $\bm{x} \in \{0, 1, 2, 3\}^n$ and $\bm{m} \cdot \bm{\theta}=\sum_{k=1}^n m_k \theta_k$. 
We rename the variables $\theta_k \to x_k\pi/2+\theta_k$ and $m_k \to 4m_k+y_k$ where $x_k$, $y_k \in \{0, 1, 2, 3\}$ and $m_k$ are their $k$-th dits, respectively.  The result is
\begin{equation}
\label{Hmin_varrho_vectorXE}
\varrho_{\bm{X}E} = \left(\frac{2}{\pi}\right)^n \int _{0}^{\frac{\pi}{2}} d\theta_1 \cdots \int _{0}^{\frac{\pi}{2}} d\theta_n  \int _{0}^{\infty} d\mu_1 \cdots  \int _{0}^{\infty} d\mu_n \  \varrho^{\bm{\theta}, \bm{\mu}}_{\bm{X}E}
\end{equation}
where
\begin{equation}
 \varrho^{\bm{\theta}, \bm{\mu}}_{\bm{X}E} := \frac{1}{4^n} \sum_{\bm{x}} \ket{\bm{x}}\bra{\bm{x}} \otimes P\{ \sum_{\bm{y}} e^{-i\frac{\pi}{2} \bm{x} \cdot \bm{y}} \sum_{\bm{m}}\sqrt{ p_{\bm{m}} \prod_{k=1}^n \frac{e^{-\mu_k}\mu^{4m_k+y_k}}{(4m_k+y_k)!}} \ e^{-i(4\bm{m}+\bm{y}) \cdot \bm{\theta}} \ket{e_{4\bm{m}+\bm{y}}} \}.
\end{equation}
This form matches that of Eq. \eqref{rho_XE} if $d=4^n$.

\begin{corollary} 
\label{Cor:lower_bound_Hmin_varrho_vectorXE}
 The min-entropy of $\varrho_{\bm{X}E}$ in Eq. \eqref{Hmin_varrho_vectorXE} obeys
\begin{equation}
H_{\min}(\bm{X}|E)_{\varrho_{\bm{X}E}} \ge 2n- \log  \left(\sum_{\bm{y}}\sqrt{q_{\bm{y}}} \right)^2.
\end{equation}
where 
\begin{equation}
q_{\bm{y}}:=\sum_{\bm{m}} p_{4\bm{m}+\bm{y}}.
\end{equation}
\end{corollary}

\begin{proof}
 This can be proven using the same logical argument in the proof of 
 Corollary~\ref{Cor:lower_bound_Hmin_varrho_XE}.
\end{proof}

We now connect the parameters $\{ q_{\bm{y}} \}$ in Corollary~\ref{Cor:lower_bound_Hmin_varrho_vectorXE} to the observable $Q$ in this SI-CV QRNG protocol and express the smoothed version in terms of $Q$.

\begin{lemma} 
\label{Lem:QRNG_smooth_H_min}
Let $\varepsilon \ge 0$. Define the set
\begin{equation}
\mathcal{S}_{\hat{Q}}:=\left\{\bm{y}\colon \frac{\sum_k \bm{y}_k }{n} \le \hat{Q}\right\},
\end{equation}
where $\bm{y}_k=0$ if $y_k=0$, $\bm{y}_k=1$ if $y_k \ge 1$, and 
\begin{equation}
\label{upper_Q}
\hat{Q} \le Q+\sqrt{\frac{(n+k)(k+1)}{nk^2}\ln \frac{1}{\varepsilon}}.
\end{equation} 
Suppose ${\sum_{\bm{y} \in \mathcal{S}_{\hat{Q}} } q_{\bm{y}}}=1-\varepsilon$. Then
\begin{equation}
\label{H_min_varepsilon}
H_{\min}^{\varepsilon}(\bm{X}|E)_{\varrho_{\bm{X}E}} \ge n[2-H(\hat{Q})],
\end{equation}
where $H(\hat{Q}):=-\hat{Q}\log \hat{Q} -(1-\hat{Q})\log[(1-\hat{Q})/3]$ is the Shannon entropy of a random variable with four possible states following the probability distribution $\{\hat{Q}, (1-\hat{Q})/3, (1-\hat{Q})/3, (1-\hat{Q})/3\}$.
\end{lemma}

\begin{proof}
The logical flow of this proof is the same as that of Lemma~\ref{Lem:smooth_H_min}.
We claim that the state
\begin{equation}
\ket{\Phi_{\varepsilon}}^n_{AE}=\frac{1} {\sqrt{{\sum_{\bm{y} \in \mathcal{S}_{\hat{Q}} } \sum_{\bm{m}} p_{4\bm{m}+\bm{y}}}}}  \sum_{\bm{y} \in \mathcal{S}_{\hat{Q}}} \sum_{\bm{m}} \sqrt{p_{4\bm{m}+\bm{y}}} \ket{4\bm{m}+\bm{y}} \otimes \ket{e_{4\bm{m}+\bm{y}}} .
\end{equation}
is $\varepsilon$-close to $\ket{\Phi}^n_{AE}$ in Eq. \eqref{Phi_n_AE}. To see this, we calculate the fidelity between the two state, that is,
\begin{equation}
F(\ket{\Phi_{\varepsilon}}^n_{AE}, \ket{\Phi}^n_{AE})=\left|\bra{\Phi_{\varepsilon}}^n_{AE} \ket{\Phi}^n_{AE}\right|=\sqrt{{\sum_{\bm{y} \in \mathcal{S}_{\hat{Q}} }q_{\bm{y}}}}=\sqrt{1-\varepsilon},
\end{equation}
so that the purified distance between the two states is then given by $\varepsilon$. 
Suppose $\varsigma_{\bm{X}E}$ is the state resulted from $\ket{\Phi_{\varepsilon}}^n_{AE}$ by Alice performing POVM $\{ \Pi_x \}$. The monotonicity of the fidelity implies that
\begin{equation}
P(\varrho_{\bm{X}E}, \varsigma_{\bm{X}E}) \le \varepsilon.
\end{equation}
Therefore, according to Corollary \ref{Cor:lower_bound_Hmin_varrho_vectorXE}, we get
\begin{equation}
H^{\varepsilon}_{\min}(\bm{X}|E)_{\varrho_{\bm{X}E}} \ge H_{\min}(\bm{X}|E)_{\varsigma_{\bm{X}E}} \ge
2n-\log  \frac{\left(\sum_{\bm{y} \in \mathcal{S}_{\hat{Q}}} \sqrt{  q_{\bm{y}}} \right)^2} {\sum_{\bm{y} \in \mathcal{S}_{\hat{Q}} } q_{\bm{y}}} .
\end{equation}
By the technique of Lemma 4 in Appendix B of Ref. \cite{wang2021tight},
we have
\begin{equation}
\log  \frac{\left(\sum_{\bm{y} \in \mathcal{S}_{\hat{Q}}} \sqrt{  q_{\bm{y}}} \right)^2} {\sum_{\bm{y} \in \mathcal{S}_{\hat{Q}} } q_{\bm{y}}} \le H(\hat{Q}).
\end{equation}
This completes the proof.
\end{proof}

Since the probability of finding a non-vacuum state in the testing rounds is $Q$, the chance of getting a Fock state whose number of particles does not equal $4m$ for some $m\in \mathbb{N}$ is at most $Q$. By the Serfling inequality, we see that the probability of such Fock state occurring is at most $\hat{Q}$ in the randomness generation rounds except a small failure probability $\varepsilon$.

Finally, we show the length of secret random numbers, and prove that this SI-CV QRNG protocol is $\varepsilon_\text{sec}$-secret with $\varepsilon_\text{sec}=4\varepsilon$

\begin{theorem} 
\label{Thrm:QRNG_keylength}
If the final random number length is given by 
\begin{equation}
\label{key_length_QRNG}
 \ell \le n[2-H(\hat{Q})]-\log\frac{1}{\varepsilon_{\normalfont \text{sec}}^2}, 
\end{equation}
 then this protocol is $\varepsilon_{\normalfont\text{sec}}$-secret. 
\end{theorem}

\begin{proof}
According to quantum leftover hashing lemma \cite{renner2008security,tomamichel2011leftover}, the users can extract a $\Delta$-secret key of length $\ell$ from string $\bm{X}$, where
\begin{equation}
\Delta=2\varepsilon+\frac{1}{2}\sqrt{2^{\ell-H_{\text{min}}^{\varepsilon}(\bm{X}|E)_{\varrho_{\bm{X}E}}}}.
\end{equation}
Then choosing $\varepsilon_\text{sec}=4\varepsilon$, we have
\begin{equation}
\Delta \le 2\varepsilon+\frac{1}{2}\sqrt{2^{\ell-H_{\text{min}}^{\varepsilon}(\bm{X}|E)_{\varrho_{\bm{X}E}}}} \le \frac{\varepsilon_\text{sec}}{2}+\frac{\varepsilon_\text{sec}}{2}=\varepsilon_\text{sec},
\end{equation}
where we use Eq. \eqref{H_min_varepsilon}. Thus, this protocol is $\varepsilon_\text{sec}$-secret.
\end{proof}

\section{Summary}

In summary, we develop a powerful technique to calculate non-trivial lower bound on min-entropy as well as its smoothed version for a given classical-quantum state by reducing the computation to a problem of eigenvalues of the adversary state.  This eases the lower bound computation.  Using three examples, we demonstrate the usefulness of our one-shot min-entropy calculation technique in computing the upper bound on the information obtained by an adversary in both discrete and continuous variable finite-data size problems in quantum key distribution and quantum random number generation.

\begin{acknowledgments}
We thank Chao Wang for the discussion of experimental feasibility of our proposed SI-CV QRNG protocol. This work is supported by the RGC Grant No. 17303323 of the HKSAR Government.
\end{acknowledgments}

\bibliographystyle{apsrev4-2}

\bibliography{one-shot}

\twocolumngrid

\end{document}